\documentclass[11pt]{article}
\usepackage{authblk}
\usepackage{soul}

\title{Practical sampling schemes for quantum phase estimation}
\author{E. van den Berg}
\affil{IBM T.J.~Watson, Yorktown Heights, NY, USA}

\usepackage[left=2cm,top=2cm,right=2cm,nohead]{geometry}
\usepackage{amsmath,amssymb,amsthm}
\usepackage{color}
\usepackage{graphicx}
\usepackage{algorithm2e}
\usepackage{multirow}
\newtheorem{theorem}{Theorem}[section]

\newcommand{\sfrac}[2]{\ensuremath{\textstyle\frac{#1}{#2}}}
\newcommand{\ket}[1]{\ensuremath{\vert{#1}\rangle}}

\newcommand{\abs}[1]{\ensuremath{\vert{#1}\vert}}

\newtheorem{innercustomgeneric}{\customgenericname}
\providecommand{\customgenericname}{}
\newcommand{\newcustomtheorem}[2]{%
  \newenvironment{#1}[1]
  {%
   \renewcommand\customgenericname{#2}%
   \renewcommand\theinnercustomgeneric{##1}%
   \innercustomgeneric
  }
  {\endinnercustomgeneric}
}

\newcustomtheorem{LabelTheorem}{Theorem}

\begin{document}

\maketitle

\begin{abstract}
  In this work we consider practical implementations of Kitaev's
  algorithm for quantum phase estimation. We analyze the use of phase
  shifts that simplify the estimation of successive bits in the
  estimation of unknown phase $\varphi$. By using increasingly
  accurate shifts we reduce the number of measurements to the point
  where only a single measurements in needed for each additional
  bit. This results in an algorithm that can estimate $\varphi$ to an
  accuracy of $2^{-(m+2)}$ with probability at least $1-\epsilon$
  using $N_{\epsilon} + m$ measurements, where $N_{\epsilon}$ is a
  constant that depends only on $\epsilon$ and the particular sampling
  algorithm. We present different sampling algorithms and study the
  exact number of measurements needed through careful numerical
  evaluation, and provide theoretical bounds and numerical values for
  $N_{\epsilon}$.
\end{abstract}

\section{Introduction}\label{Sec:Introduction}

Given a unitary operator $U$ and one of its eigenvectors $\ket{\xi}$,
we would like to obtain an accurate estimate of the corresponding
eigenvalue $\lambda$, which, as a result of unitarity, can be written
as $\lambda = e^{i2\pi \varphi}$ with $\varphi \in [0,1)$.  The
quantum phase estimation problem considers finding an estimate of the
phase $\varphi$, from which the estimate of $\lambda$ is then easily
found. The importance of quantum phase estimation is highlighted by
the wide range of applications that rely on it, including Shor's prime
factorization algorithm~\cite{SHO1997a}, quantum
chemistry~\cite{ASP2005DLHa,WHI2011BAa,OMA2016BKRa}, and quantum
Metropolis sampling~\cite{TEM2011OVPa}.

\begin{figure}[t]
\centering
\includegraphics[width=.38\textwidth]{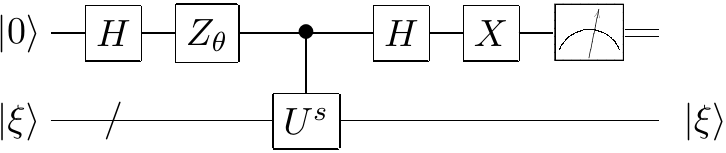}
\caption{Quantum circuit for phase estimation of $U^s$.}\label{Fig:Circuit}
\end{figure}

The two main approaches to quantum phase estimation are the
Fourier-based approach described
in~\cite{CLE1998EMMa,KIT1995a,NIE2010Ca}, and Kitaev's
algorithm~\cite{KIT1995a,KIT2002SVa}, which we study in this work. The
quantum circuit central to Kitaev's algorithm is illustrated in
Figure~\ref{Fig:Circuit}. Assuming an initial state of
$\ket{0}\ket{\xi}$, the circuit first applies a Hadamard operation on
the first qubit, followed by a phase-shift operation
\[
Z_{\theta} = \left[\begin{matrix}1 & 0 \\ 0 & e^{i2\pi \theta}\end{matrix}\right].
\]
The circuit then applies $U^s$ on the $\ket{\xi}$ qubits, conditioned
on the first qubit, where
\begin{equation}\label{Eq:Us}
U^s\ket{\xi} = \lambda^s\ket{\xi} = e^{i2\pi (s\varphi)}\ket{\xi}.
\end{equation}
Finally, a second Hadamard operation is applied on the first qubit
followed by a {\sc{not}} gate\footnote{The {\sc{not}} gate is inserted
  only to simplify exposition, as it allows us to focus on
  measurements with value 1 instead of 0, which will be more
  natural. In practice, the measurements could be negated. All
  algorithms presented in this paper apply equivalently to non-negated
  measurements with minor changes.} and measurement. The circuit
performs the following mapping:
\[
\ket{0}\otimes\ket{\xi} \rightarrow \left(\frac{1 - e^{i2\pi (s\varphi
    + \theta)}}{2}\ket{0} + \frac{1 + e^{i2\pi (s\varphi
    + \theta)}}{2}\ket{1}\right)\otimes\ket{\xi}
\]
Measuring the first qubit therefore returns 1 with probability
\[
P_{\theta}(1 \mid \varphi) = \left\vert \frac{1+ e^{i2\pi(\varphi + \theta)}}{2}\right\vert^2
= \frac{1 + \cos(2\pi (\varphi + \theta))}{2},
\]
and 0 otherwise. By appropriately choosing $\theta$ we can obtain
information on the sine and cosine of $\varphi$, since
\begin{equation}\label{Eq:CosSin}
\cos(2\pi\varphi) = 2P_{0}(1 \mid \varphi) - 1,\quad\mbox{and}\quad
\sin(2\pi\varphi) = 2P_{-1/4}(1 \mid \varphi) - 1.
\end{equation}
Repeated measurements of the quantum circuit with $\theta=0$ and
$\theta=-1/4$ allow us to approximate $P_0$ and $P_{-1/4}$, from which
the sine and cosine values and subsequently $\tilde{\varphi}$ can then
be determined.

Throughout this work we use the convention that angles $\varphi$,
$\theta$, and $\omega$ are expressed in radians divided by $2\pi$,
whereas angles $\alpha$ are always expressed in radians. We frequently
use
\[
p_x(\alpha) = \frac{1 + \cos(\alpha)}{2},\quad\mathrm{and}\quad
p_y(\alpha) = \frac{1 + \sin(\alpha)}{2},
\]
and write $p_x$ and $p_y$ when the dependency on $\alpha$ is
clear.

\section{Kitaev's algorithm}

The goal in Kitaev's algorithm~\cite{KIT2002SVa} is to obtain the
approximation
\[
\tilde{\varphi} = .\overline{\beta'_1\beta'_2\cdots
  \beta'_{m}\beta'_{m+1} \beta'_{m+2}} = \sum_{j=1}^{m+2}2^{-j}\beta'_j,
\]
for $\varphi = .\overline{\beta_1\beta_2\beta_3\cdots}$, such that
$\vert \varphi - \tilde{\varphi}\vert \leq 2^{-(m+2)}$ holds with
probability at least $1 - \epsilon$. The key principle behind the
algorithm lies in the fact that using $U^s$ in \eqref{Eq:Us} with
$s=2^{j-1}$ gives measurements about $s\varphi$, and therefore amounts
to shifting the bits in $\varphi$ to the left by $j-1$ positions:
\[
\varphi_j := 2^{j-1}\varphi \equiv
.\overline{\beta_{j}\beta_{j+1}\beta_{j+2}\cdots}\ \mathrm{mod}\ 1.
\]
The multiplicative factor of $2\pi$ causes the phase to be invariant
to the integer part $\overline{\beta_{1}\cdots\beta_{j-1}}$, which can
therefore be omitted. As a result, by choosing $j$ we can work with
the bitstring starting from any $\beta_{j}$. Using this, the first
step of Kitaev's algorithm is to choose $j=m$ and estimate
$\tilde{\varphi}_j = .\overline{\beta'_{m}\beta'_{m+1}\beta'_{m+2}}$
such that $\vert\varphi_j - \tilde{\varphi}_j\vert \leq 1/8$ with
probability at least $1 - \bar{\epsilon}_m$. This is done by taking
sufficiently many samples (discussed in detail in
Section~\ref{Sec:Sampling}) to estimate $\cos(2\pi \varphi_j)$ and
$\sin(2\pi \varphi_j)$ such that the approximated angle $\omega_j$
deviates from $\varphi_j$ by at most $1/16$ with probability at least
$1-\bar{\epsilon}_m$. The angle $\omega_j$ is then quantized (rounded)
to the nearest integer multiple of $1/8$ modulo 1 to obtain
$\tilde{\varphi}_j$. With a maximum quantization error of $1/16$, this
amounts to an accuracy of at least $1/8$.

The second stage of the algorithm iteratively adds one new bit
$\beta'_{j}$ per step, for $j=m-1,\ldots,1$. For each step we first
obtain a $1/16$ accurate approximation $\omega_j$ of $\varphi_j$ with
probability at least $1-\bar{\epsilon}_j$ using the technique outlined
above. Next, we enforce consistency using the bitstring known so far
and set
\[
\beta'_{j} := \begin{cases}%
0 & \mathrm{if}\ \vert.\overline{0\beta'_{j+1}\beta'_{j+2}} -
\mbox{$\omega_j$}\vert < 1/4,\\
1 & \mathrm{if}\ \vert.\overline{1\beta'_{j+1}\beta'_{j+2}} -
\mbox{$\omega_j$}\vert < 1/4.\end{cases}
\]
Using a union bound on the error probabilities, it can be seen that
the algorithm succeeds with probability at least
$1 - \sum_{j=1}^m\bar{\epsilon}_j$. Choosing
$\epsilon_{j} = \epsilon/m$ we can therefore ensure with probability
at least $1-\epsilon$ that all approximate angles $\omega_j$ are
valid, in which case $\tilde{\varphi}$ will be $2^{-(m+2)}$
accurate. A graphical illustration of Kitaev's algorithm is given in
Figure~\ref{Fig:Kitaev}.

\begin{figure}
\centering
\begin{tabular}{cc}
\includegraphics[width=0.4\textwidth]{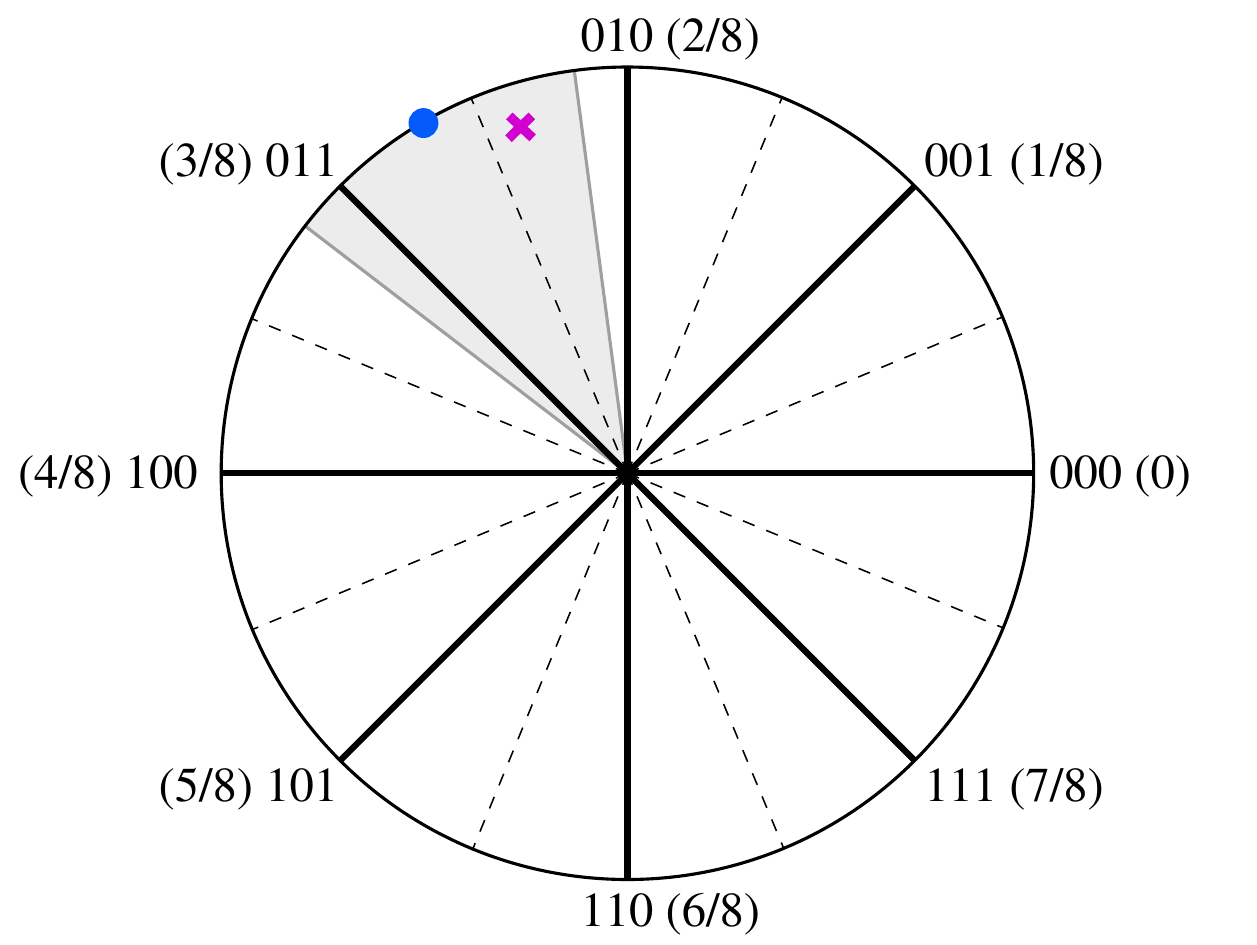}&
\includegraphics[width=0.4\textwidth]{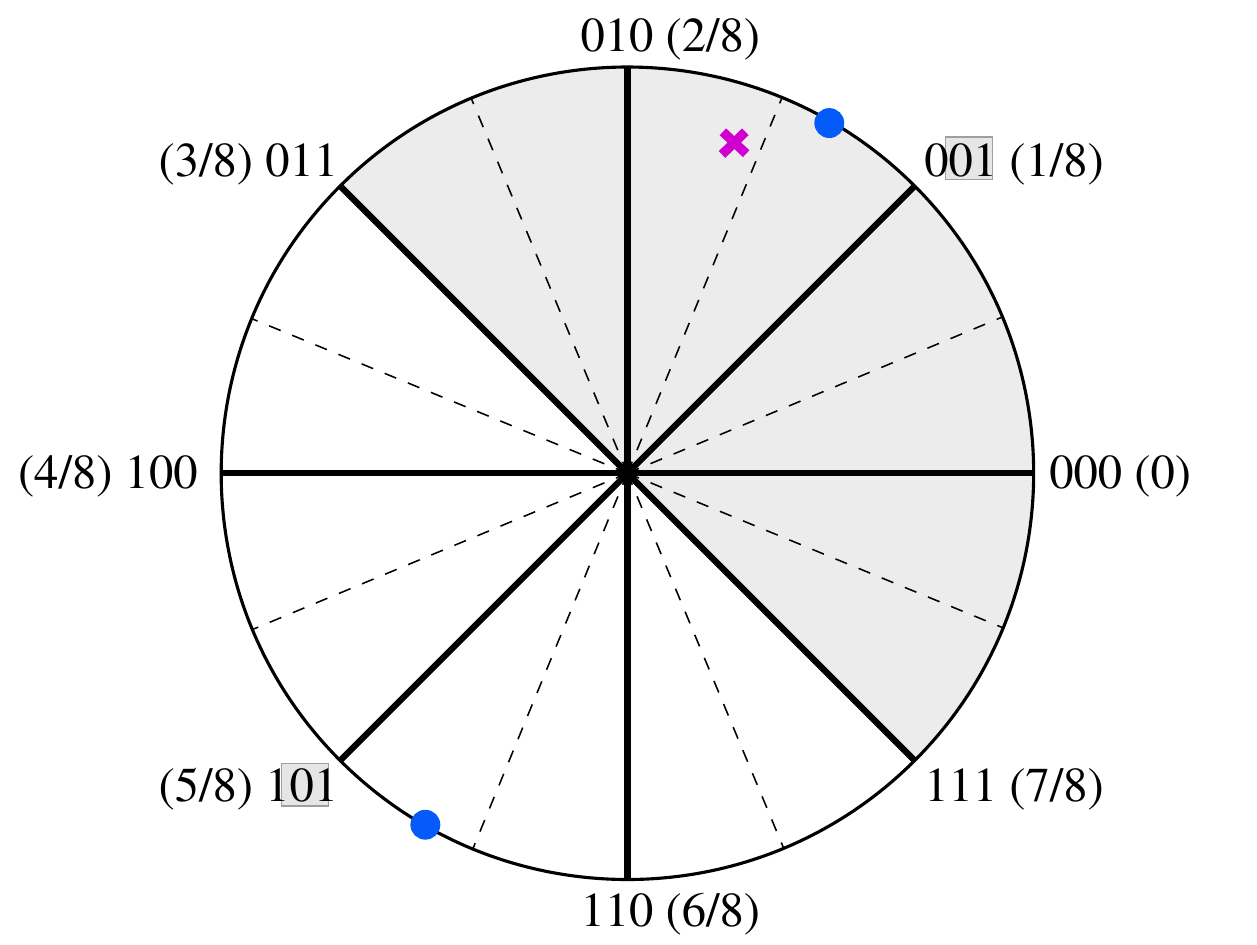}\\
({\bf{a}}) & ({\bf{b}})
\end{tabular}
\caption{({\bf{a}}) Estimation of the angle $\phi_j$ indicated by the
  blue dot is done by obtaining estimates of the sine and cosine
  values (indicated by the purple cross), from which an approximate
  angle $\omega_j$ and quantization ($\underline{01}$0) can be
  obtained. Angle $\phi_j$ is equal to $2\phi_{j-1}$ modulo 1, and
  could therefore originate from either of the blue dots in plot
  ({\bf{b}}). Fixing the last two bits to the leading two from the
  previous estimates gives points closest to $0\underline{01}$,
  indicated by the shaded region, or closest to
  $1\underline{01}$. Measuring the purple cross determines
  $\beta'_j = 0$.  }\label{Fig:Kitaev}
\end{figure}

The estimation of the sine or cosine terms in \eqref{Eq:CosSin} with
accuracy $\delta$ requires estimation of the probability terms to
accuracy $\delta/2$. Let $s_1,\ldots,s_n$, be i.i.d.\ samples of a
Bernoulli distribution with probability of success
$p = 1 - P_{\theta}(1)$. Denoting
$\tilde{p}_n = (1/n)\sum_{i=1}^n s_i$, then it follows from the
Chernoff bound that
\begin{equation}\label{Eq:Chernoff}
\mathrm{Pr}[\vert p - \tilde{p}_n\vert \geq \delta/2] \leq 2e^{-\delta^2n/2}.
\end{equation}
We require that the probability $2e^{-\delta^2n/2}$ be bounded by
$\bar{\epsilon}$, which is guaranteed for
\begin{equation}\label{Eq:MChernoff}
n \geq n(\delta,\epsilon) =\frac{2}{\delta^2} \log(2/\bar{\epsilon}).
\end{equation}
For the theoretical complexity of Kitaev's algorithm, note that $m$
angle estimations $\omega_j$ are used, each requiring approximate sine
and cosine values, therefore yielding a total of $2m$ estimations. By
choosing $\bar{\epsilon} = \epsilon / 2m$ we obtain an overall sample
complexity of
$\frac{4m}{\delta^2}\log(4m/\epsilon) =
\mathcal{O}(m\log(m/\epsilon))$.

\section{Improvements using phase shifts}\label{Sec:PhaseShift}

As a practical improvement to steps in the second stage of Kitaev's
algorithm, consider the situation shown in
Figure~\ref{Fig:Kitaev}(b). The exact value of $\varphi_j$ is equal to
either $(0 + \varphi_{j+1})/2$, or $(1+\varphi_{j+1})/2$, as indicated
by the blue dots. Using the first two bits of the binary
representation for $\tilde{\varphi}_{j+1}$, in this case 01, we
therefore know that $\varphi_{j}$ is at most 1/8 away from either 001
or 101. Using this information, we can first rotate $\varphi_{j}$ by
the reference angle, in this case 001, to obtain a new angle that is
1/8 away from either 000 or 100, as illustrated in
Figure~\ref{Fig:Rotation}(a). Now, instead of approximating both sine
and cosine we only need to determine the sign of the cosine, which
requires far fewer measurements. We then set $\beta'_j$ based on the
sign: if it is positive we set $\beta'_j = 0$, otherwise we set
$\beta'_j = 1$. We can further improve this scheme by maintaining all
known bits and rotate by 0010, instead of by the truncated version
001. Doing so we obtain an angle that is now at most 1/16 away from 0
or 1/2, as shown in Figure~\ref{Fig:Rotation}(b). By using the full
binary string $\tilde{\varphi}_{j+1}/2$ at each stage, we get
increasingly small deviations from 0 or 1/2, which increases the
magnitude of the cosine value and reduces the number of measurements
needed to accurately determine the correct sign.

\paragraph{Increasingly accurate rotations}

The use of existing measurements to correct or zero out portions of
the binary representation in iterative phase estimation was described
earlier by~\cite{CHI2000PRa,KNI2007OSa} along with its connection to
phase estimation based on the quantum Fourier transformation. Here we
analyze the measurement complexity in detail (see
also~\cite{DOB2007JSWa} for a high-level analysis). Given that the
bits $\beta'_j$ are determined in reverse order (that is, from
  least to most significant), we switch to working with iterations,
such that iteration $k \in [1,m]$ determines $\beta'_{m+1-k}$.  We now
consider iterations $k\geq 2$, which comprise the second stage of the
algorithm.  At each iteration in the second stage we need to determine
the sign of the cosine of the shifted angle. In terms of measurements,
this amounts to sampling a majority of either zeros or ones.  For the
number of measurements we have the following result.

\begin{figure}[t!]
\centering
\begin{tabular}{cc}
\includegraphics[width=0.4\textwidth]{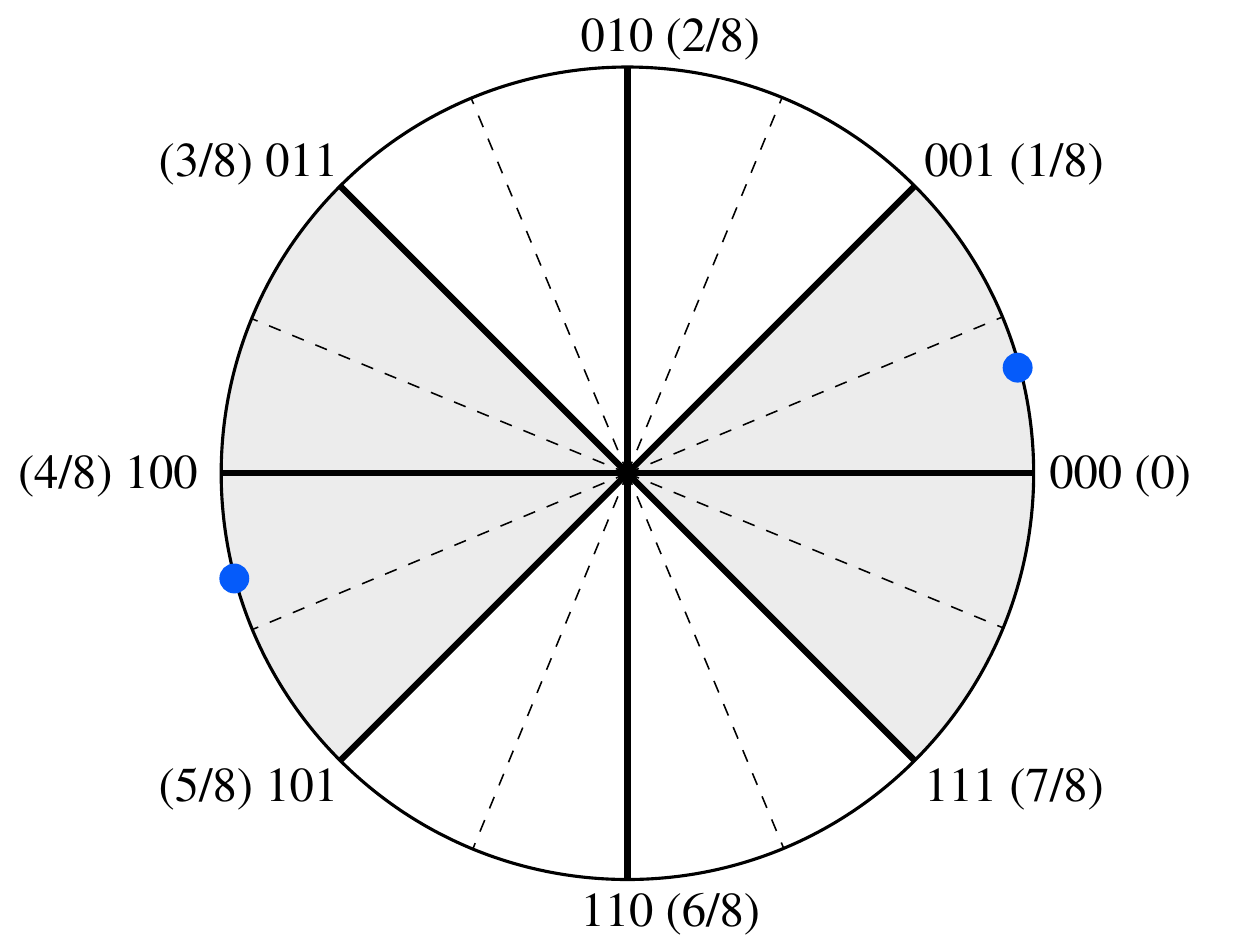}&
\includegraphics[width=0.4\textwidth]{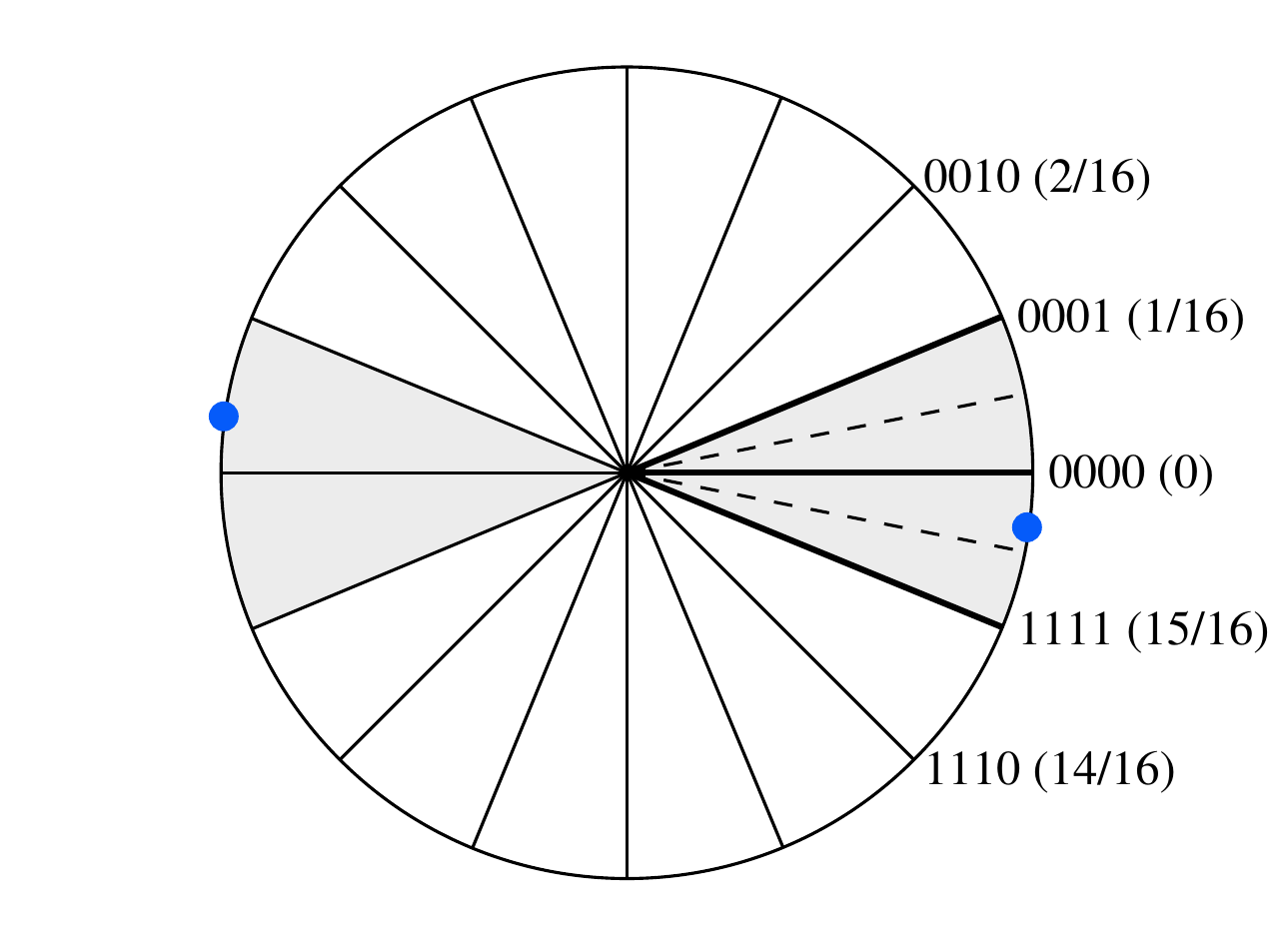}\\
({\bf{a}}) & ({\bf{b}})
\end{tabular}
\caption{Application of a phase shift ensures that the unknown angle
  (indicated by the blue dot) lies within one of the highlighted
  wedges centered on the horizontal axis. The more accurate the
  reference angle, the smaller the wedge, as shown in (a) with maximum
  deviation $1/8$, and (b) $1/16$.}\label{Fig:Rotation}
\end{figure}

\begin{theorem}\label{Thm:MeasurementsSign}
  Let $\alpha \in (0,\pi/2)$. Then we can correctly distinguish angles
  from the sets $[-\alpha,\alpha]$ and $[\pi-\alpha,\pi+\alpha]$ with
  probability at least $1-\epsilon$ by checking whether the majority
  of $n$ measurements is 1 or 0, whenever
\begin{equation}\label{Eq:MeasurementsSign}
n \geq  \frac{\log(1/\epsilon)}{\log(1/\sin(\alpha))}.
\end{equation}
\end{theorem}
\begin{proof}
  Assume that unknown angle lies in the range $[-\alpha,\alpha]$ and
  denote by $p = (1 + \cos(\alpha))/2$. The probability that at most
  $k$ out of $n$ measurements are 1, with $k < np$ is bounded
  by~\cite{ARR1989Ga}:
\begin{equation}\label{Eq:SumBound}
\mathrm{Pr}(X \leq k) \leq \exp\left(-nD(k/n \Vert p)\right),
\end{equation}
where $D(a \Vert p)$ denotes the relative entropy
\[
D(a \Vert p) = a\log\left(\frac{a}{p}\right) + (1-a)\log\left(\frac{1-a}{1-p}\right).
\]
We want the majority of the measurements to be 1 and an error
therefore occurs whenever $k \leq n/2$. Choosing $k=n/2$, which is
allowed since $p>1/2$, gives $a = 1/2$ and an error $\epsilon'$ bounded by
\begin{equation}\label{Eq:SignEpsilon}
\epsilon' \leq \exp\left(-\frac{n}{2}\log(1/(4p(1-p)))\right)
\end{equation}
To simplify, note that
\begin{equation}\label{Eq:4p(1-p)}
4p(1-p) = 4\left(\frac{1 + \cos(\alpha)}{2}\right)\left(\frac{1 - \cos(\alpha)}{2}\right) = 1 -
\cos^2(\alpha) = \sin^2(\alpha).
\end{equation}
We want to ensure that the right-hand side of \eqref{Eq:SignEpsilon}
is less than or equal to $\epsilon$. Taking logarithms and simplification then
gives the desired result. The result for $[\pi-\alpha,\pi+\alpha]$
follows similarly.
\end{proof}

At iteration $k$ we apply a phase shift based on the angle from the
previous iteration, and Theorem~\ref{Thm:MeasurementsSign} therefore
applies with $\alpha$ equal to $\alpha_k = \pi/2^{k+1}$ (as an
example, at iteration $k=2$ we have a maximum deviation of 1/16 or
$\pi/8$). When taking a single measurement, the probability of failure
is $1-p_k$ where $p_k = p_x(\alpha_k)$. In the special case where $k$
is such that $1 - p_k \leq \bar{\epsilon}$, it is therefore follows
that only a single sample is needed.  It holds that
\begin{eqnarray}
1 - p_k 
&=& (1 - \cos(\pi/2^{k+1})) / 2\notag\\
&=& \sin^2(\pi/2^{k+2})\notag\\
&\leq& \pi^2/2^{2(k+2)},\label{Eq:Bound_1-pk}
\end{eqnarray}
where we use the identity $1 - \cos(x) = 2\sin^2(x/2)$ in the second
line and $\sin(x) \leq x$ for $x\geq 0$ in the last. We want to find
the value of $k$ such that iterations $k$ through $m$ each require
only a single measurement and have a combined error bounded by
$\bar{\epsilon}$.  Choosing $\bar{\epsilon} = \epsilon/k$, gives the requirement
\[
\sum_{j=k}^{m}(1-p_j) \leq \epsilon/k
\]
Bounding the left-hand side as
\begin{equation}\label{Eq:UnlimitedIteration}
\sum_{j=k}^{m} (1-p_j) \leq \sum_{j=k}^{m}\pi^2 / 2^{2(j+2)} \leq
\pi^2 \sum_{j=k}^{\infty} (1/4)^{j+2} =
\frac{4\pi^2}{3}\left(\frac{1}{4}\right)^{k+2} = \frac{\pi^2}{12}4^{-k},
\end{equation}
we obtain the sufficient condition
\begin{equation}\label{Eq:SufficientK}
4^{-k} \leq \frac{12\epsilon}{k\pi^2}.
\end{equation}
Taking base-two logarithm and rearranging gives
\[
2k -  \log_2(k/12) \geq \log_2(\pi^2 / \epsilon).
\]
It can be verified that $\log_2(k/12) < k/22$ for $k\geq 0$, and we
can therefore choose
\begin{equation}\label{Eq:BoundK}
k \geq k_\epsilon := \left\lceil\frac{22}{43}\log_2(\pi^2/\epsilon)\right\rceil\!.
\end{equation}

The value of $k_{\epsilon}$ does not depend on $m$ and satisfies that
for $k_{\epsilon} \geq 2$ for all $\epsilon \in (0,1]$. The bound on
the sum of errors for iterations $k \geq k_{\epsilon}$
in~\eqref{Eq:UnlimitedIteration} can be seen to apply for any $m$.
Denote by $N_{\epsilon}$ the total number of measurements taken in the
first $k_{\epsilon}-1$ iterations, each with an error not exceeding
$\bar{\epsilon}$.  When $m < k_{\epsilon}$ it is clear that at most
$N_{\epsilon}$ samples are needed. When $m \geq k_{\epsilon}$ we need
to take an additional sample for each of the remaining
$m - k_{\epsilon}+1$ steps. The overall sampling complexity is
therefore bounded by $N_{\epsilon} + m$, where $N_{\epsilon}$ depends
only on $\epsilon$ and the sampling methods used for the first
$k_{\epsilon}-1$ iterations. The work by Svore {\it{et
    al.}}~\cite{SVO2014HFa} proposes a phase estimation algorithm with
  complexity $\mathcal{O}(m\log^*(m))$. Unlike the proposed method,
  however, their algorihm allows parallelization and the use of
  clusters, and does not require arbitrarily accurate phase shifts,
  which can be expensive from a circuit perspective (see the
  Discussion section for more details).

\section{Practical sampling schemes}\label{Sec:Sampling}

In this section we consider different sampling schemes and study the
number of samples needed to attain a desired accuracy with a given
error rate. For the evaluation of the error rate, we consider the
measurements as a Binomial random variable by counting the number of
successful measurements (which could be either 0 or 1, depending on
the context). In order to determine the angle to a certain accuracy
using $n$ measurements, the number of successful measurements $X$
typically needs to fall in some set $\mathcal{K} \subseteq [n]$, where
$[n]$ denotes the set $\{0,1,\ldots.n\}$. For $p = p_x(\alpha)$ this is satisfied with
probability
\begin{equation}\label{Eq:Binomial1}
\mathrm{Pr}(\mathcal{K} \mid n,\alpha) = \sum_{k \in\mathcal{K}} {n \choose k}p^k(1-p)^{n-k}.
\end{equation}
We also require two-dimensional settings with probabilities $p_x$ and
$p_y$ and $\mathcal{K}\in[n]\times[n]$, given by
\begin{equation}\label{Eq:Binomial2}
\mathrm{Pr}(\mathcal{K} \mid n,\alpha) = \sum_{(k_x,k_y) \in
  \mathcal{K}} {n \choose k_x}{n \choose k_y}p_x^{k_x}(1-p_x)^{n-k_x}p_y^{k_y}(1-p_y)^{n-k_y}.
\end{equation}
The error rate is then given by $1-\mathrm{Pr}(\mathcal{K})$, and the
goal is to find the minimum $n$ for which the error is below some
threshold $\bar{\epsilon}$. In addition to providing bounds, we use
the GNU multi-precision arithmetic
library\footnote{\texttt{http://gmplib.org/}} to evaluate the
probabilities in~\eqref{Eq:Binomial1} and~\eqref{Eq:Binomial2}
numerically, thus allowing us to find the exact minimum number of
samples needed to attain the desired error level for each of the
methods. The best sampling schemes are then used in
Section~\ref{Sec:Nepsilon} to obtain numerical
values for $N_{\epsilon}$.

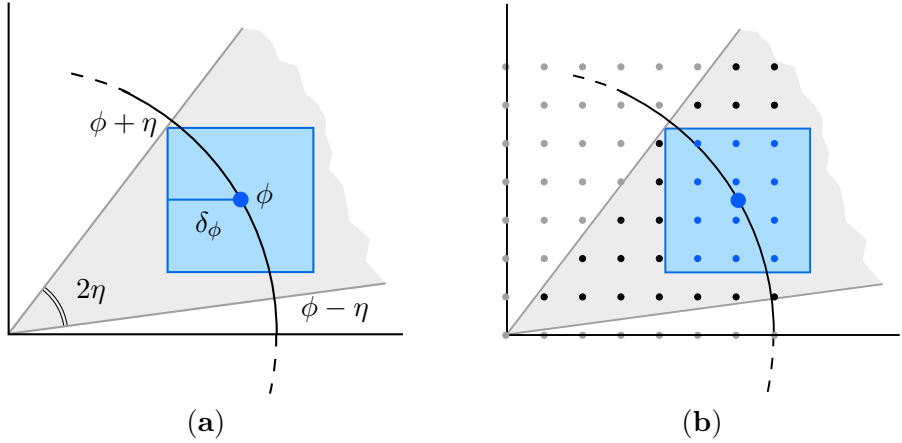
\begin{figure}[t]
\centering
\begin{tabular}{ccc}
\includegraphics[height=150pt]{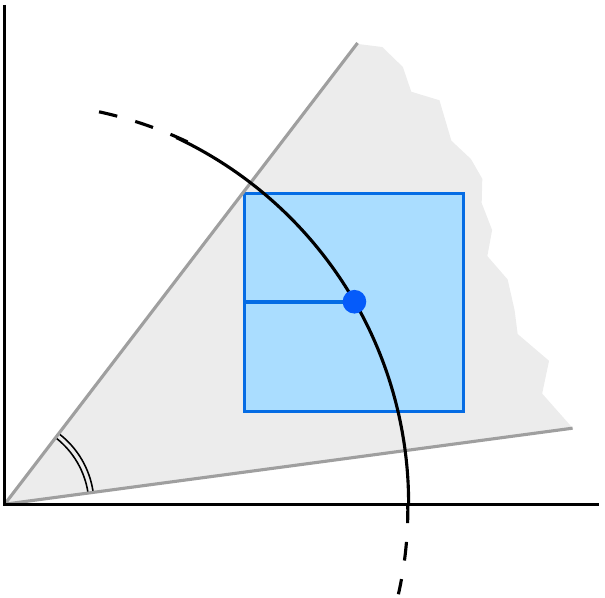}\begin{picture}(0,0)
\put(-125,37){$2\eta$}
\put(-80,62){$\delta_{\phi}$}
\put(-57,73){$\phi$}
\put(-40,30){$\phi-\eta$}
\put(-120,100){$\phi+\eta$}
\end{picture}& \hspace*{12pt}&
\includegraphics[height=150pt]{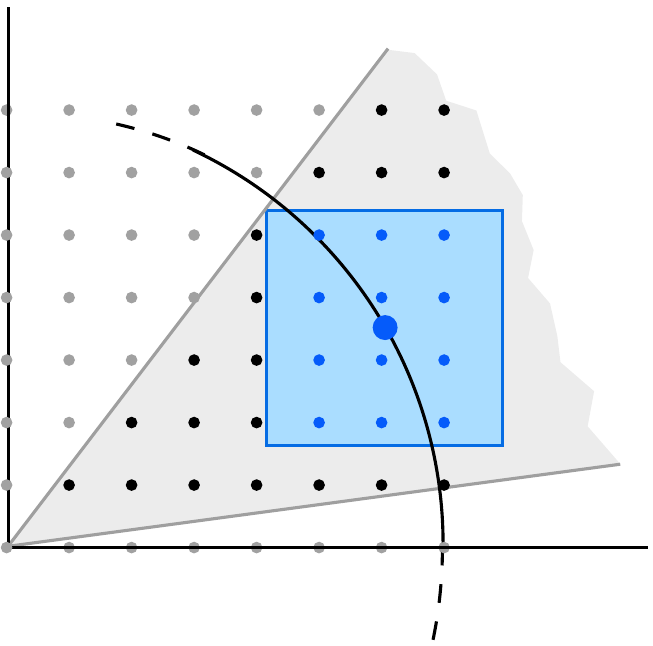}\\
({\bf{a}}) && ({\bf{b}})
\end{tabular}
\caption{Estimation of angle $\phi$ with accuracy $\eta$ (a) is
  guaranteed when the sine and cosine terms are each estimated with
  precision $\delta_{\phi}$, where the subscript denotes the
  dependence on $\phi$; (b) the same estimation based on discrete
  points during sampling, in this case using 15 samples each for sine
  and cosine (remaining three orthants omitted).}\label{Fig:Delta}
\end{figure}

\subsection{Box-based sine and cosine}\label{Sec:BoxBased}

\begin{figure}
\centering
\begin{tabular}{cc}
\includegraphics[width=0.475\textwidth]{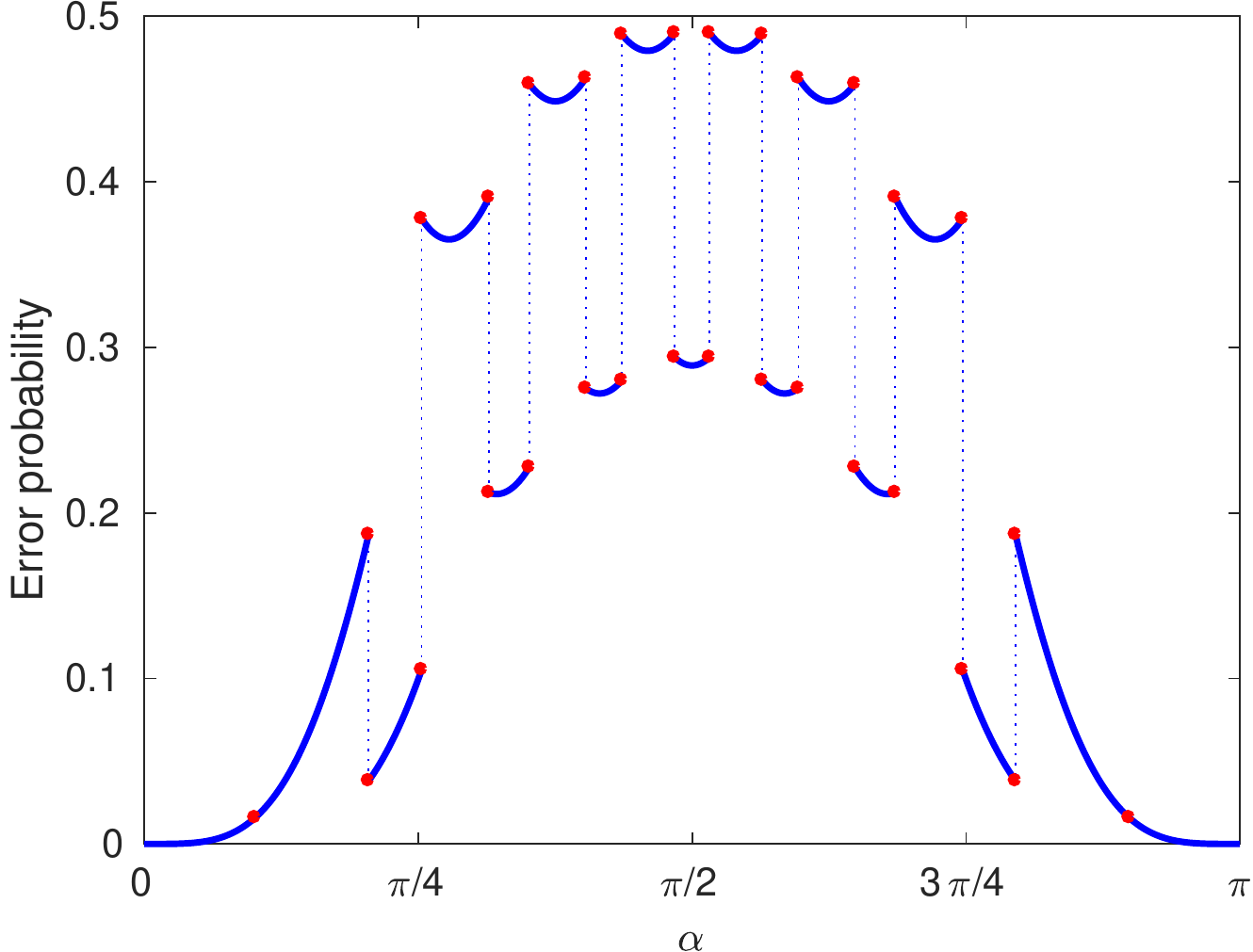}&
\includegraphics[width=0.475\textwidth]{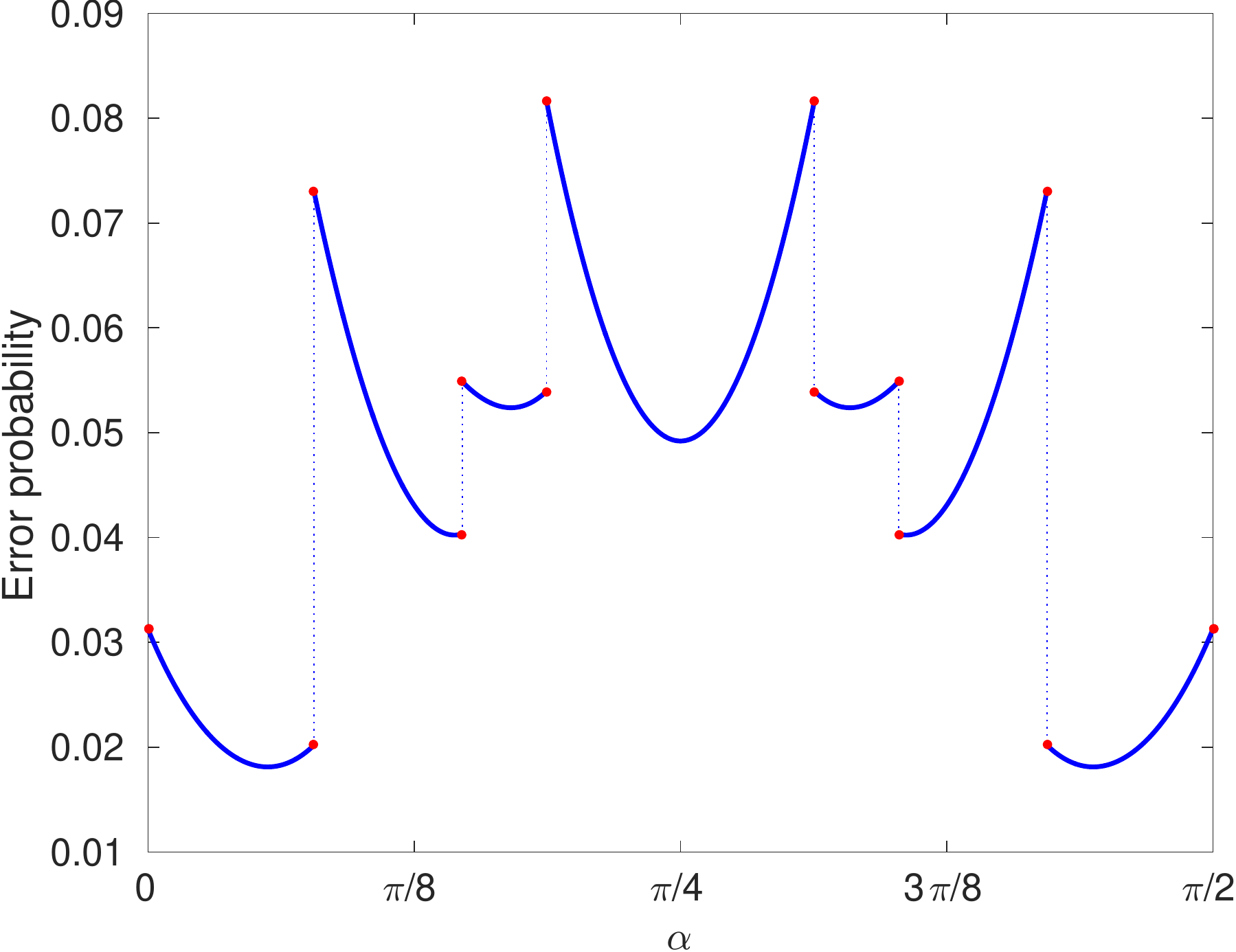}\\
({\bf{a}}) & ({\bf{b}})
\end{tabular}
\caption{Error probability for ({\bf{a}}) estimating cosine
  for different angles for $n=8$ using a bounding box with $\delta =
  0.3$; ({\bf{b}}) error probability for sampling different angles
  using the wedge-based approach with $n=5$ and $\eta=\pi/4$.}\label{Fig:ErrorJumps}
\end{figure}

The first sampling method we look at the is box-based scheme discussed
in \cite[Section 13.5.2]{KIT2002SVa} and illustrated in
Figure~\ref{Fig:Delta}(a). The idea is to independently estimate
$\cos(2\pi\varphi)$ and $\sin(2\pi\varphi)$ to an accuracy $\eta$,
such that the recovered angle differs no more than $\eta$ (in radians)
from the actual angle, with probability at least $1-\epsilon$. The
following theorem gives the maximum deviation $\delta(\eta)$ allowed
in the sine and cosine estimates to reach the desired accuracy in the
angle (a proof of the theorem is given in Appendix~\ref{Sec:ProofThmDelta}):

\begin{theorem}\label{Thm:Delta}
  For any $0 \leq \eta \leq \pi/2$ we can compute an estimate
  $\tilde{\phi}$ of any $\phi \in [0,2\pi]$ with accuracy
  $\vert \tilde{\phi} - \phi\vert \leq \eta$ from sine and cosine
  estimates $\tilde{c}$ and $\tilde{s}$ with $\vert \tilde{c} -
  \cos(\phi)\vert \leq \delta$ and $\vert\tilde{s}
  -\sin(\phi)\vert \leq \delta$, whenever 
\begin{equation}
\delta \leq \delta(\mbox{$\eta$}) = \frac{\sin(\eta)}{\sqrt{2}}.
\end{equation}
For uniform estimation over $\phi$ this bound is tight.
\end{theorem}

Estimating the cosine is equivalent to estimating the probability
$p = (1 + \cos(2\pi\varphi))/2$ with accuracy $\delta(\eta)/2$. When
taking $n$ measurements we can estimate the probability as $k/n$,
where $k$ is the number of measurements that are 1. The success set
$\mathcal{K}_{n,\delta}(p)$ is therefore defined as
\begin{equation}\label{Eq:SetKnp}
\mathcal{K}_{n,\delta}(p) := \{k\in[n] \mid (p-\delta/2)n \leq k\leq (p+\delta/2)n\},
\end{equation}
from which we can then evaluate
$\mathrm{Pr}(\mathcal{K}_{n,\delta}(p))$. One difficulty here is that
the probability $p$ depends on the unknown angle $2\pi\varphi$ and we
therefore need to consider the error rate for all possible
angles. Figure~\ref{Fig:ErrorJumps}(a) illustrates the error
probability $1-\mathrm{Pr}(\mathcal{K}_{n,\delta}(p))$ as a function
of angle for $\delta=0.3$, when taking eight measurements. Due to the
discrete nature of the samples there are numerous discontinuities,
which are best explained using Figure~\ref{Fig:Delta}(b). Consider the
box centered around a point on the circle at a given angle. For the
cosine we only need to consider the horizontal component and it may
therefore help to think of a projection of the box onto the horizontal
axis. Given an angle and corresponding probability $p$, the set
$\mathcal{K}_{n,p}$ then consists of all the points on the horizontal
axis within the projected box. As the angle increases, the box shifts,
thereby gradually adding and removing points from the set
$\mathcal{K}$ at critical angles. When looking at the limit as the
angle approaches a critical angle, we can take $p$ to be the
probability associated with the critical angle. It then follows from
\eqref{Eq:Binomial1} that the probability of success increases when a
point is added to the set, and decreases when a point is removed, and
vice versa for the error rate. Indeed, these discontinuities are
clearly seen in the error rate plotted in
Figure~\ref{Fig:ErrorJumps}(a).  The following theorem, which we prove
appendix~\ref{Sec:ProofThmBoxConvex}, shows that for sufficiently
large $n$, the error curve is piecewise convex in $p$:

\begin{theorem}\label{Thm:BoxConvex}
  Choose $\delta > 0$ and let $f_n(p) = 1 - \mathrm{Pr}(X \in
  \mathcal{K}_{n,\delta}(p))$. Then for $n \geq \max\{1 +
  1/\delta^2,3\}$, $f_n(p)$ is piecewise convex on $[0,1]$ with breakpoints at
  $[0,1] \cap \{(k/n) \pm \delta/2\}_{k\in[n]}$.
\end{theorem}

In order to find the maximum error it therefore suffices to evaluate
the error function at the critical angles with boundary points removed
from $\mathcal{K}_{n,\delta}$. Note that the lower bound on $n$ is a
sufficient condition, and Figure~\ref{Fig:ErrorJumps}(a) indicates
that the condition on $n$ may be improved or eliminated.

\paragraph{Multi-stage evaluation.}
The measurement scheme described in \cite{KIT2002SVa} determines
$\varphi$ with accuracy $1/16$ and then quantizes to three bits, which
adds a maximum deviation of $1/16$, to obtain the desired $1/8$
accurate approximation. Instead of attempting to determine the angle
in a single pass, we can also apply the general idea behind Kitaev's
algorithm and use a multi-level approach. For a two-level approach we
start with a $1/4$-accurate estimate for $2\varphi$, using $1/8$
accurate angle estimation followed by two-bit quantization
$.\overline{\beta'_1\beta'_2}$. Based on this, we know that
$.\overline{0\beta'_1\beta'_2}$ or $.\overline{1\beta'_1\beta'_2}$ is
a $1/8$ accurate approximation for $\varphi$ and therefore only need
to determine the leading bit, which is conveniently done using the
phase-shift technique described in Section~\ref{Sec:PhaseShift}.  Even
though we quantize the estimation of $2\varphi$ to two bits, we can
decide how accurately we want to represent the unquantized estimation,
obtained based on the sine and cosine estimates, for use in the phase
shift. Using $k \geq 2$ bits gives a maximum deviation of
$1/8 + 1/2^{k+1}$, which, after halving, gives an angle
$\pi/8 + \pi/2^{k+1}$ for the determination of the sign of the
cosine. As shown in Table~\ref{Table:Sign}, the smaller the angle the
fewer measurements are needed to attain a desired confidence level.

For a three-stage approach we estimate the unquantized angle
$4\varphi$ with accuracy $1/4$, and then quantize to a single bit. We
then apply two stages of sign determination using a phase shift based
on the unquantized angle or a $k$-bit discretization to obtain the
final $1/8$ accurate three-bit quantized estimate for $\varphi$.

\begin{table}
\centering
\begin{tabular}{lrrrrrrrrrr}
\hline
{\bf{Angle}} $/$ $\epsilon$ & $10^{-1}$ & $10^{-2}$ & $10^{-3}$ & $10^{-4}$ & $10^{-5}$ & $10^{-6}$ & $10^{-7}$ & $10^{-8}$ & $10^{-9}$ & $10^{-10}$\\
\hline
$7\pi/16$ & 43 & 139 & 247 & 357 & 469 & 583 & 697 & 813 & 927 & 1043\\
$6\pi/16$ & 11 & 35 & 61 & 87 & 115 & 143 & 171 & 199 & 227 & 257\\
$5\pi/16$ & 5 & 15 & 27 & 37 & 49 & 61 & 73 & 85 & 97 & 111\\
$4\pi/16$ & 3 & 9 & 15 & 21 & 27 & 33 & 39 & 45 & 53 & 59\\
$3\pi/16$ & 1 & 5 & 9 & 13 & 15 & 19 & 23 & 27 & 31 & 35\\
$2\pi/16$ & 1 & 3 & 5 & 7 & 9 & 13 & 15 & 17 & 19 & 21\\
$\pi/16$ & 1 & 1 & 3 & 5 & 5 & 7 & 9 & 9 & 11 & 13\\
$\pi/32$ & 1 & 1 & 3 & 3 & 5 & 5 & 7 & 7 & 9 & 9\\
$\pi/64$ & 1 & 1 & 1 & 3 & 3 & 5 & 5 & 5 & 7 & 7\\
$\pi/128$ & 1 & 1 & 1 & 3 & 3 & 3 & 3 & 5 & 5 & 5\\
$\pi/256$ & 1 & 1 & 1 & 1 & 3 & 3 & 3 & 3 & 5 & 5\\
\hline
\end{tabular}
\caption{Number of measurement requires to correctly determine the
  sign of $\cos(\alpha)$ with probability at least $1-\epsilon$, when
  $\alpha$ deviates from $0$ or $\pi$ by at most the given angle.}\label{Table:Sign}
\end{table}

\begin{table}[t]
\centering
\begin{tabular}{lrrrrrrrrrr}
\hline
{\bf{Description}} $/$ $\epsilon$ & $10^{-1}$ & $10^{-2}$ & $10^{-3}$ & $10^{-4}$ & $10^{-5}$ & $10^{-6}$ & $10^{-7}$ & $10^{-8}$ & $10^{-9}$ & $10^{-10}$\\
\hline
Single-stage box                    &  112 &  222 &  334 &  452 &  570 &  688 &  806 &  932 & 1050 & 1176\\
Two-stage box (2-bits)              &   45 &   83 &  121 &  159 &  201 &  241 &  283 &  325 &  365 &  407\\
Two-stage box (3-bits)              &   43 &   79 &  115 &  149 &  189 &  225 &  265 &  305 &  343 &  383\\
Two-stage box (exact)               &   43 &   77 &  111 &  145 &  183 &  217 &  255 &  293 &  331 &  369\\
Three-stage box (2-bits)            &   52 &   96 &  142 &  190 &  240 &  286 &  336 &  386 &  434 &  484\\
Three-stage box (3-bits)            &   38 &   64 &   96 &  126 &  160 &  190 &  224 &  256 &  288 &  320\\
Three-stage box (exact)             &   32 &   54 &   78 &  104 &  130 &  154 &  180 &  208 &  234 &  260\\
Single-stage box, jointly           &   82 &  186 &  296 &  414 &  534 &  652 &  778 &  896 & 1014 & 1140\\
Single-stage wedge                  &   40 &   84 &  132 &  188 &  244 &  300 &  358 &  416 &  476 &  534\\
Two-stage wedge (2-bit)             &   19 &   37 &   57 &   77 &   99 &  119 &  141 &  163 &  185 &  207\\
Two-stage wedge (3-bit)             &   17 &   33 &   49 &   67 &   87 &  105 &  125 &  145 &  163 &  183\\
Two-stage wedge (exact)             &   15 &   29 &   47 &   63 &   81 &   97 &  115 &  133 &  149 &  169\\
Three-stage wedge (2-bit)           &   30 &   66 &  102 &  138 &  176 &  216 &  254 &  292 &  330 &  368\\
Three-stage wedge (3-bit)           &   18 &   38 &   58 &   78 &   98 &  122 &  142 &  164 &  184 &  206\\
Three-stage wedge (exact)           &   14 &   28 &   42 &   56 &   70 &   88 &  102 &  116 &  132 &  146\\
Sign based                          &   15 &   33 &   51 &   69 &   87 &  105 &  123 &  141 &  165 &  183\\
Sign based (bound)                  & {\it{  28}} & {\it{  48}} & {\it{  68}} & {\it{  88}} & {\it{ 108}} & {\it{ 128}} & {\it{ 148}} & {\it{ 168}} & {\it{ 188}} & {\it{ 208}}\\
Majority and sign                   &   17 &   29 &   41 &   55 &   67 &   79 &   93 &  105 &  119 &  133\\
Majority and sign (bound)           & {\it{  22}} & {\it{  35}} & {\it{  48}} & {\it{  62}} & {\it{  75}} & {\it{  88}} & {\it{ 102}} & {\it{ 115}} & {\it{ 128}} & {\it{ 141}}\\
\hline
\end{tabular}
\caption{Number of measurements required using different methods to
  obtain a $1/8$-accurate quantized estimate of $\varphi$ with probability at least
  $1-\epsilon$.}\label{Table:MeasurementsFirstStage}
\end{table}

\paragraph{Numerical evaluation.}
We numerically evaluate the different box-based schemes and summarize
the number of measurements for different error rates $\epsilon$ in
Table~\ref{Table:MeasurementsFirstStage}. For the single-stage
measurement scheme we choose $\bar{\epsilon} = \epsilon/2$ to estimate
the sine and cosine values. For the two-stage scheme we use
$\epsilon/4$ for the sine and cosine estimation and $\epsilon/2$ for
the second stage, while for the three-stage scheme we use $\epsilon/6$
for sine and cosine, and $\epsilon/3$ for the last two stages. For
each of the instances the condition on $n$ in
Theorem~\ref{Thm:BoxConvex} is satisfied and the numbers reported are
therefore optimal for the given setting. Some reduction in the number
of measurements in the box-based measurement schemes may however still
be obtained by partitioning $\epsilon$ differently over the different
stages. The best results are obtained with a three-stage approach, due
to the reduction in the number of samples required for the box-based
part, as well the limited number of samples needed to accurately
determine the sign of the cosine (see Table~\ref{Table:Sign}). As a
final remark, note that adding more measurements can temporarily
increase the error rate due to the discrete nature underlying the
error curve. As an example, it can be shown that approximation of the
cosine with $\delta=0.1$ succeeds with probability at least $0.7$ for
all $n \geq 110$ except $n\in\{113,\ldots,119\}$. Even though these
transition regions are not always present, they do show that care
needs to be taken when changing the number of samples.

\paragraph{Joint determination of sine and cosine error.} For the box
sampling scheme, we determine the number of samples required based on
the maximum error probability of the cosine component over all
angles. The same number of measurements is then used to independently
estimate the sine component. The sine error curve is the same as the
cosine error curve, but shifted by $\pi/2$. Based on
Figure~\ref{Fig:ErrorJumps}(a) we see that the error probabilities
tend to complement, and that joint determination of the error should
therefore help reduce the number of measurements. As shown in
Table~\ref{Table:MeasurementsFirstStage} (Single stage box, jointly),
this is indeed the case, but the reduction is somewhat modest, and in
fact, the reduced number of measurements can be no smaller than that
based on the cosine error probability evaluated with
$\bar{\epsilon} = \epsilon$ rather than $\epsilon/2$.

\subsection{Wedge-based angle}

Given $n$ measurements for both $p_x = (1+\cos(\alpha))/2$, and
$p_y = (1+\sin(\alpha))/2$, we can denote by $n_x$ and $n_y$ the
number of 1 measurements. The box-based approach requires that
$\abs{p_y - n_y/n} \leq \delta/2$ with high probability, and likewise
for $p_x$ and $n_x/n$. This requirement enables the use of the
Chernoff bound to derive a bound on the number of samples, but is
otherwise too restrictive. From Figure~\ref{Fig:Delta}(b) we see that
accurate determination of the angle only requires that $(n_x, n_y)$
lie within a wedge with angles $\alpha - \eta$ and $\alpha + \eta$
centered at $(n/2, n/2)$. The probability of success then consists of
the set of points $\mathcal{K}_{n,\eta}(\alpha)$ within the wedge,
which strictly includes the points accepted for the box-based
approach. The number of samples required to guarantee a $1-\epsilon$
probability of success will therefore be at least as good or smaller
than the box-based approach. (Theoretical results on the error
probability in the special case of $\eta = \pi/2$ can be found
in~\cite{KIM2015LYa}.)

\begin{figure}[!h]
\centering
\begin{tabular}{cc}
\includegraphics[width=0.3\textwidth]{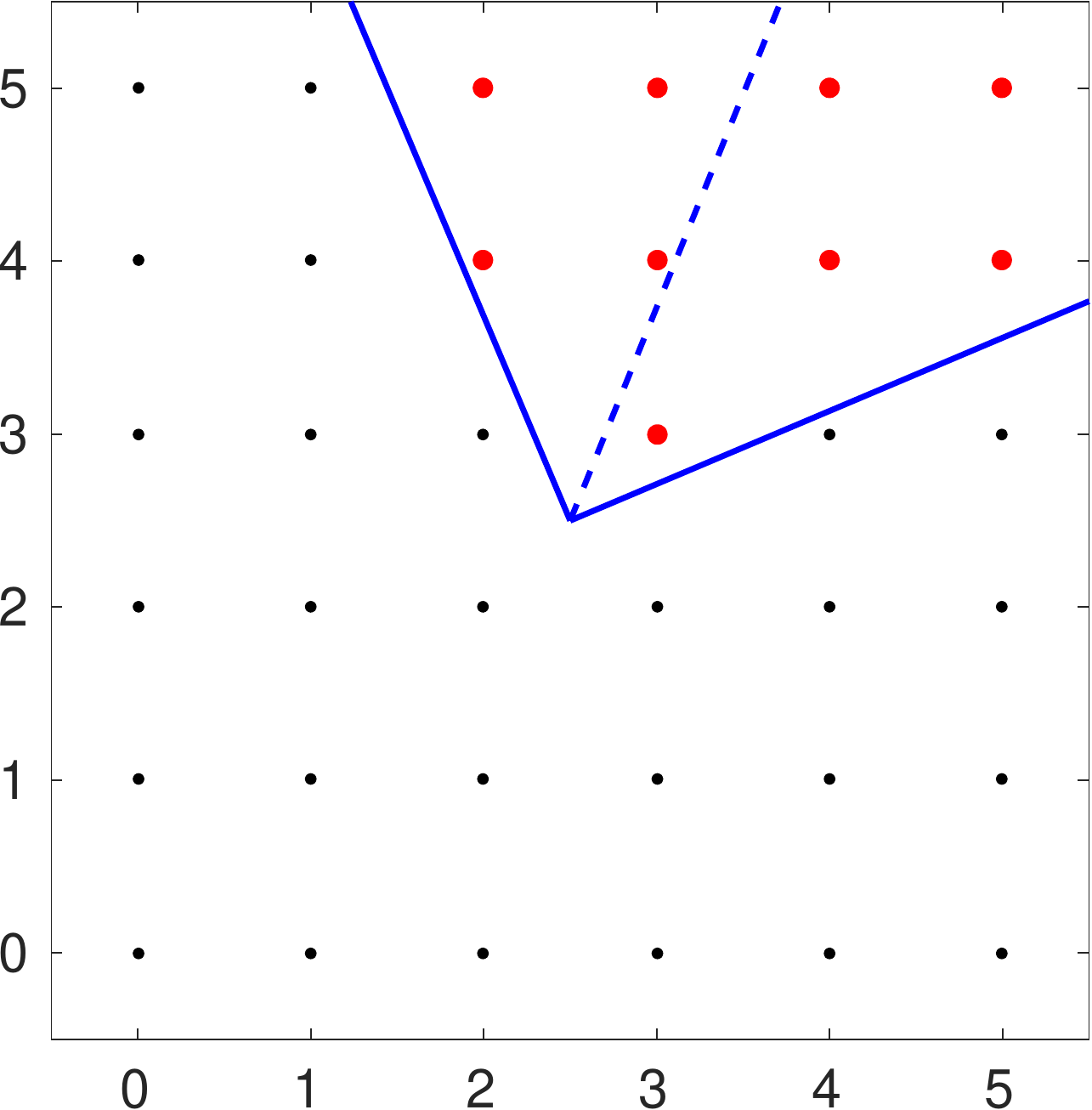}&
\includegraphics[width=0.3\textwidth]{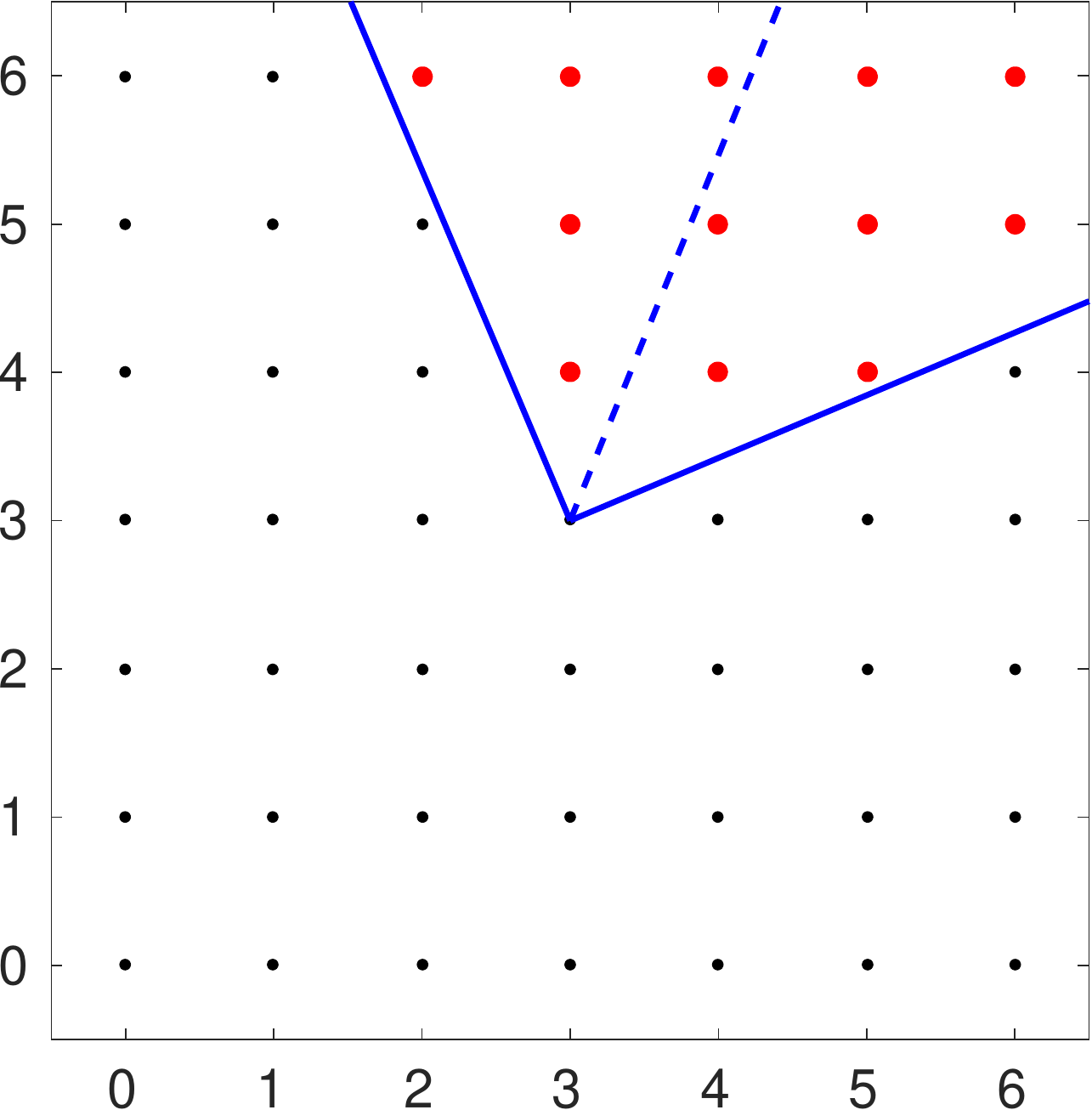}\\
({\bf{a}}) $n=5$ & ({\bf{b}}) $n=6$
\end{tabular}
\caption{Combinations of sine and cosine measurements that estimate
  the angle at the center of the wedge with maximum deviation $1/8$.}\label{Fig:WedgeGrid}
\end{figure}

\paragraph{Numerical evaluation.} Similar to the numerical evaluation
of the box approach, in order to find the error rate corresponding to
a given number of measurements $n$, we need to minimize
$\mathrm{Pr}(\mathcal{K}_{n,\eta}(\alpha))$ over $\alpha$. This can be
done by sweeping over the angles $\alpha$, determining
$\mathcal{K}_{n,\eta}(\alpha)$ at each point, as illustrated in
Figure~\ref{Fig:WedgeGrid}. An example of the resulting error
probability for $n=5$ and $\eta=\pi/4$ is shown in
Figure~\ref{Fig:ErrorJumps}(b). The break-points in the curve happen
at angles where grid points are on the boundary of the wedge (for even
$n$, the origin of the wedge $(n/2, n/2)$ is excluded and is therefore
not considered to lie on the boundary). The approach therefore is to
find all angles $\alpha$ at which the wedge boundary intersects grid
points, and evaluate the error probability at those angles, with
boundary points at either the bottom or top edges omitted to obtain
the limit as $\alpha$ approaches the critical angle from a clockwise
or counter-clockwise direction. The error probability in
Figure~\ref{Fig:ErrorJumps}(b) appears piecewise convex, but we did
not attempt to rigorously establish this. The number of measurements
determined using the above algorithm should therefore be interpreted
as a lower bound for the method. The resulting number of measurements
for the single-stage wedge-based approach, as well as the extension to
the two- and three-stage approach described in
Section~\ref{Sec:BoxBased}, are listed in
Table~\ref{Table:MeasurementsFirstStage}.

\subsection{Triple-sign sampling}

When angle $\alpha$ is at most $\pi/4$ from either $0$ or $\pi$, as
illustrated in Figure~\ref{Fig:TripleSign}(a), the sign-based approach
can be used with $p_x$ to determine with probability at least
$1-\bar{\epsilon}$ whether the point lies on the right or left of the
vertical axis. Similarly, when $\alpha$ is $\pi/4$ close to $\pi/2$ or
$3\pi/2$, as shown in Figure~\ref{Fig:TripleSign}(b), we can tell with
the same probability whether the point lies above or below the
horizontal axis. When $\alpha$ lies outside of the given range we make
no assumption on the results in either case. When applying both
schemes we can combine the obtained signs, as shown in
Figure~\ref{Fig:TripleSign}(c).  By construction we know that at least
one of the two is correct, up to the desired success rate of
$1-\bar{\epsilon}$.  For angles between $\pi/4$ and $3\pi/4$, as
illustrated in the plot, this means that when the method succeeds we
obtain either $01$ or $11$. The maximum error obtained in this case is
therefore $\pi/2$. A more convenient quantization can be obtained by
applying a phase shift of $\pi/4$ prior to applying the two
measurement steps, followed by the inverse phase shift to the
result. After changing the labels we obtain the quantization values
given in Figure~\ref{Fig:TripleSign}(d). To obtain a $1/8$ accurate
estimation we can apply an additional Kitaev step with sign-based
sampling.

\begin{figure}
\centering\setlength{\tabcolsep}{3pt}
\begin{tabular}{cccc}
\includegraphics[width=0.23\textwidth]{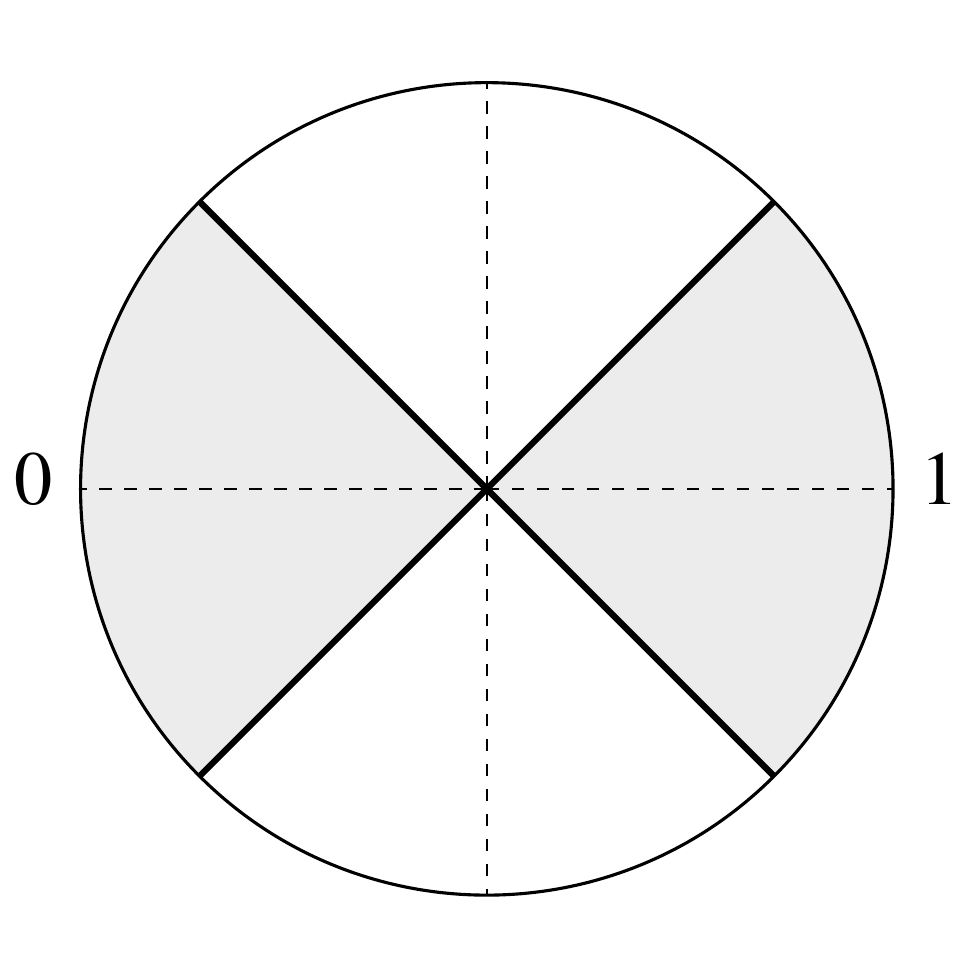}&
\includegraphics[width=0.23\textwidth]{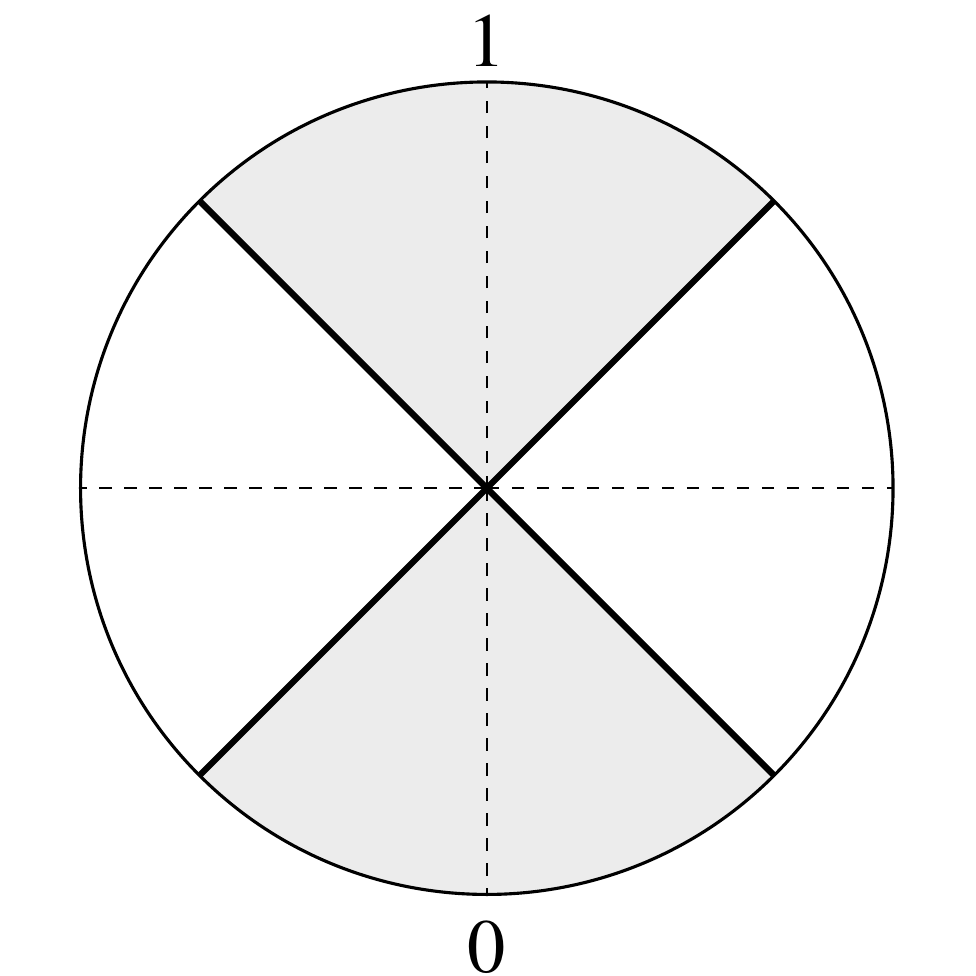}&
\includegraphics[width=0.23\textwidth]{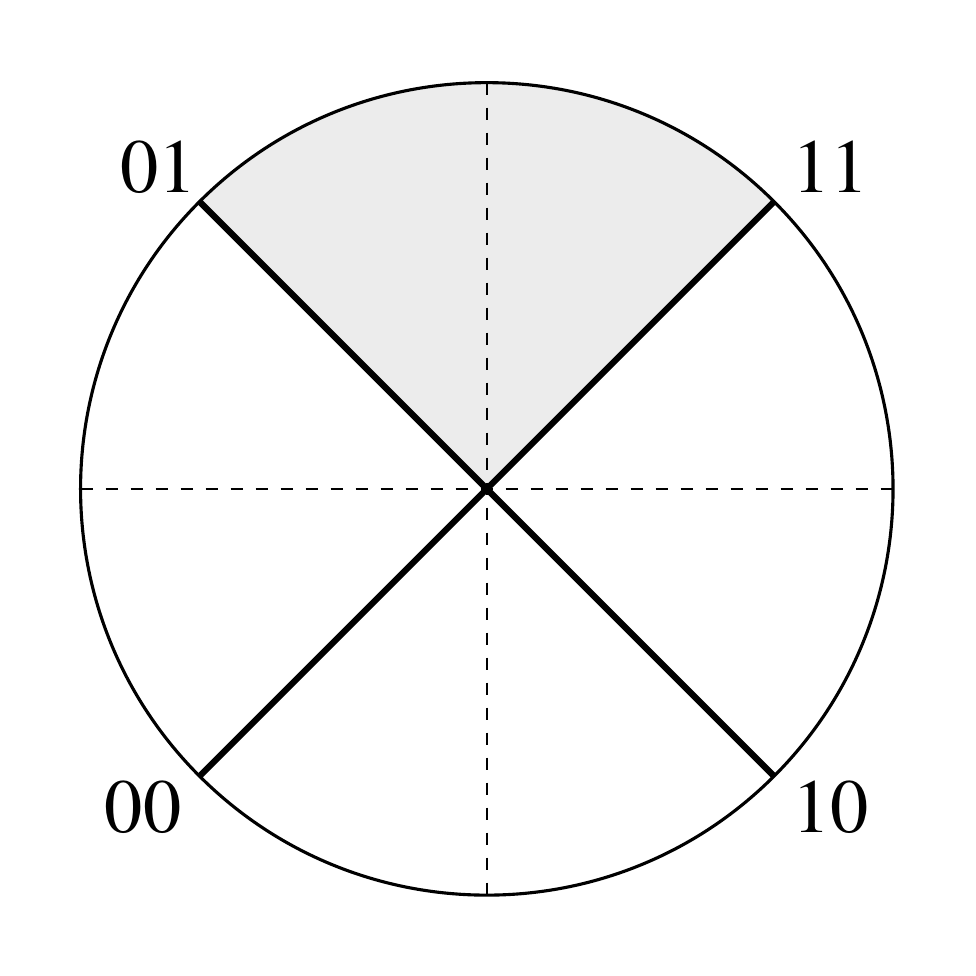}&
\includegraphics[width=0.23\textwidth]{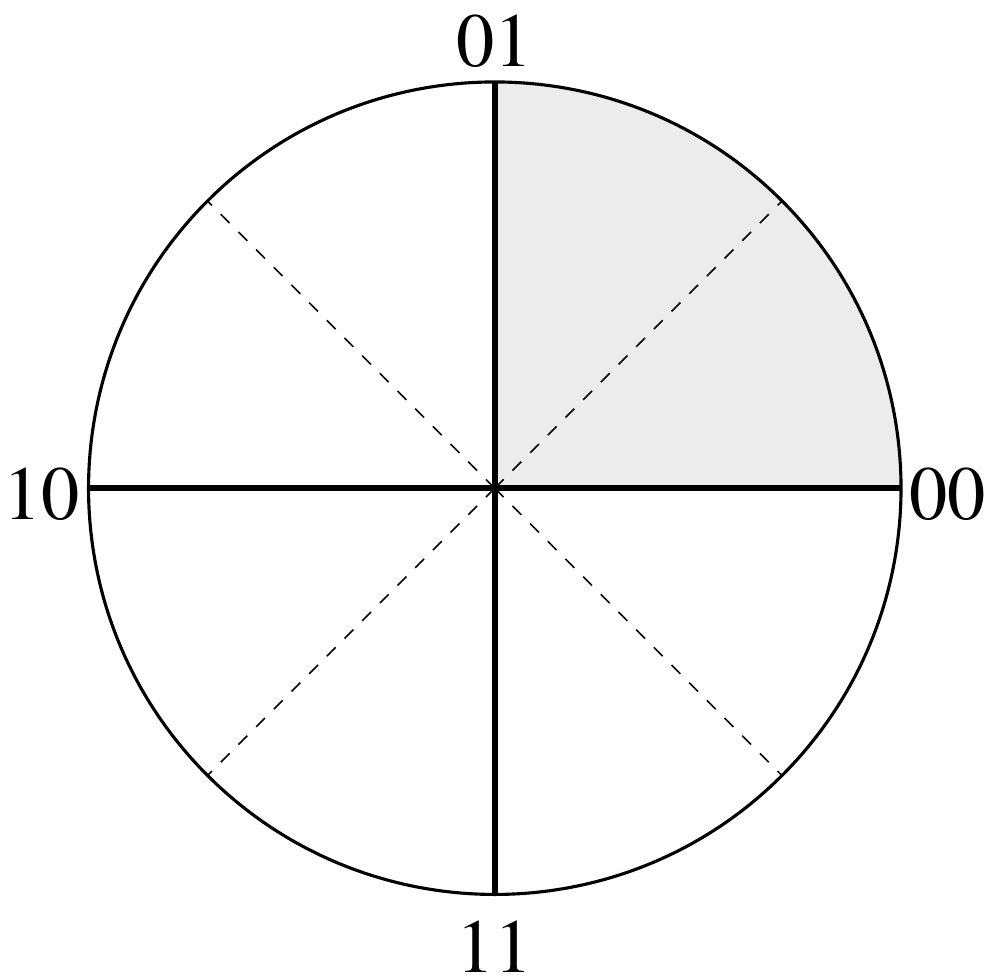}\\
({\bf{a}}) & ({\bf{b}}) & ({\bf{c}}) & ({\bf{c}})
\end{tabular}
\caption{Estimation of (a) the horizontal and (b) the vertical
  component of the angle, accurate for angles in the shaded regions;
(c) combination of the two sign-based estimates to obtain a $1/4$
accurate quantization; and (d) alternative quantization obtained by applying a $\pi/4$
rotation before the measurements and an inverse rotation and
relabeling afterwards.}\label{Fig:TripleSign}
\end{figure}

For the sufficient number of measurements per component,
Theorem~\ref{Thm:MeasurementsSign} applied with $\alpha = \pi/4$ gives
\begin{equation}\label{Eq:SignSamplesPi4}
n \geq \frac{\log(1/\epsilon)}{\log(1/\sin(\pi/4))} =
\frac{\log(1/\epsilon)}{\log(2/\sqrt{2})}
= \frac{\log(1/\epsilon)}{\log(2) - \log(2)/2}
=  2\log_2(1/\epsilon).
\end{equation}
The first stage requires $n$ measurements each for the horizontal and
vertical component. The second stage requires another $n$
measurements, for a total of $3n$ measurements. For each of the three
steps we can choose $\bar{\epsilon} = \epsilon /2$ for a total maximum
error of $\epsilon$. Note that the maximum error in the first stage is
$\bar{\epsilon}$, since one of the two components is irrelevant
(although we do not know which of the two it is). Combined we can take
\begin{equation}\label{Eq:NTripleSign}
N = 3\lceil 2\log_2(2/\epsilon)\rceil \leq 9 + 6\log_2(1/\epsilon),
\end{equation}
where the inequality is due to the addition of 3 to account for
rounding to integers.

\subsection{Majority sampling}

\begin{figure}[t]
\centering
\begin{tabular}{cccc}
$.\overline{01}$ && &\\
$.\overline{10}$
  $\vcenter{\hbox{\includegraphics[width=0.225\textwidth]{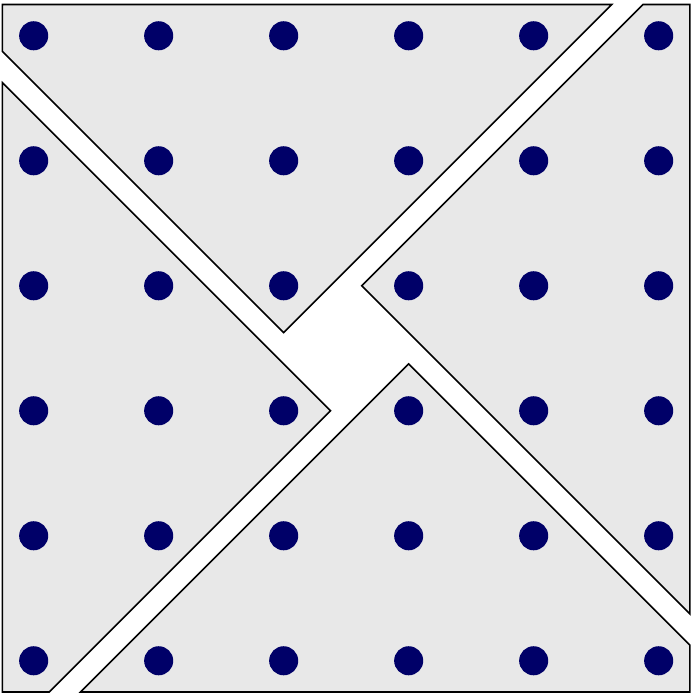}}}$ $.\overline{00}$
   &
$\vcenter{\hbox{\includegraphics[width=0.225\textwidth]{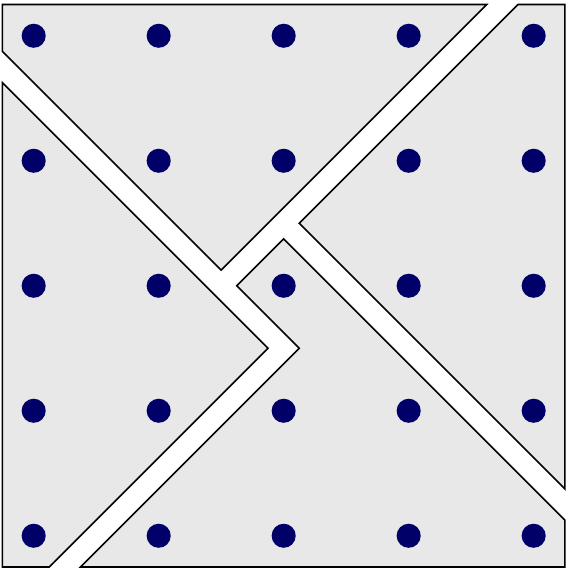}}}$
&
$\vcenter{\hbox{\rotatebox{90}{\begin{minipage}{105pt}$j\!=\!0$\hfill$\ldots$\hfill
        $n$\end{minipage}}}}$\hspace*{-12pt}&%
$\vcenter{\hbox{\includegraphics[width=0.225\textwidth]{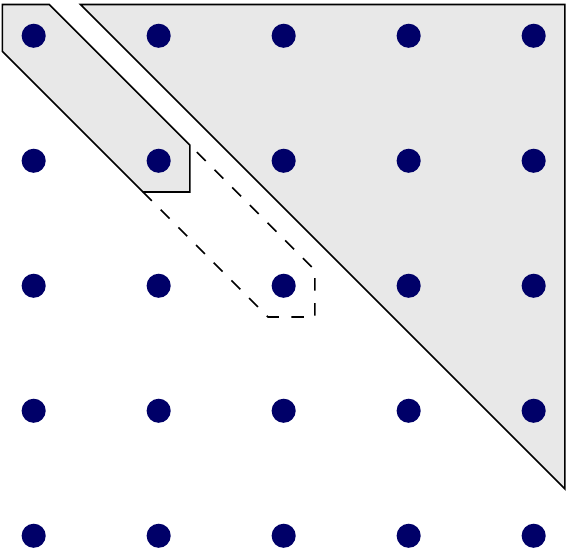}}}$
  \\
\\[-10pt]
$.\overline{11}$ & & &$i=0$\hfill$\ldots$\hfill$n$ \\
({\bf{a}}) & ({\bf{b}}) && ({\bf{c}})
\end{tabular}
\caption{Classification regions for majority-based sampling for (a) odd, and
  (b) even $n$; and (c) the success set $\mathcal{K}_n$ for angles in the range $[0,\pi/2)$.}\label{Fig:Majority}
\end{figure}

For majority sampling we take $n$ measurements for the sine and cosine
components, and count the number of positive measurements by $n_y$ and
$n_x$, respectively. The quantized approximation of the angle is
defined in terms of the majority of the number of 0 or 1
measurements
\[
q(n_x,n_y) = \begin{cases} .\overline{00} & n_x \geq \max\{n_y, n - n_y+1\} \\
.\overline{01} & n_y \geq \max\{n_x + 1, n - n_x\}\\
.\overline{10} & n-n_x \geq \max\{n_y+1, n-n_y\}\\
.\overline{11} & \mathrm{otherwise},
\end{cases}
\]
which gives partitions as illustrated in Figures~\ref{Fig:Majority}(a)
and (b).  We want to obtain an estimator that is 1/4 accurate with
probability at least $1-\epsilon$. In particular we allow angles
$\phi \in [0,1/4)$ to be quantized as either $.\overline{00}$ or
$.\overline{01}$, and similarly for interval increments of
$1/4$. Denoting by $\mathcal{K}_{ab}$ the set of points that map to
$.\overline{ab}$, this gives a success set of
$\mathcal{K}_{00} \cup \mathcal{K}_{01}$. For the analysis of the
error we work with a reduced set
$\mathcal{K}_n = \{(i,j) \mid i,j \in [0,n], j \geq n-i+1\}$,
illustrated by the top-right triangle in
Figure~\ref{Fig:Majority}(c). Based on this we have the following
result (proven in Appendix~\ref{Sec:ProofThmMajorityBound}):
\begin{theorem}\label{Thm:MajorityBound}
Let $\mathcal{K}_n = \{(i,j) \mid i,j \in [0,n], j \geq n-i+1\}$, then
for all $\alpha \in [0,\pi/2]$
\[
1 - \mathrm{Pr}(\mathcal{K}_n \mid n,\alpha) \leq \frac{2}{2^n}.
\]
\end{theorem}
\noindent We expect that this result can be improved by a factor of two. Indeed,
defining the error probability
\begin{equation}\label{Eq:fn}
f_n(\alpha) = 1 - \mathrm{Pr}(\mathcal{K}_n \mid n, \alpha)
\end{equation}
over angles $\alpha = 2\pi\varphi$, it can be seen from
Figure~\ref{Fig:MajorityError} that the error curves are convex and
attain the maximum at $\alpha = 0$. The error probability is the
summation of the probabilities for $(i,j) \not\in\mathcal{K}_n$. At
$\alpha = 0$, we have $p_x = 1$, which implies that all terms
including $(1-p_x)^{n-i}$ are zero, except those with $i=n$. The only
such point is $(n,0)$, and we therefore have
\[
f_n(0) = (1-p_y)^n = (1-\sfrac{1}{2})^n = 2^{-n}.
\]
We can now expand $\mathcal{K}_n$ with any of the points in the gray
part of the diagonal in Figure~\ref{Fig:Majority}(c). This lowers the
error for $\alpha > 0$ but does not affect $f_n(0)$. Under the
assumption that the maximum of $f_n(\alpha)$ is attained at
$\alpha = 0$, the maximum error for the set
$\mathcal{K}_{00} \cup \mathcal{K}_{01}$ is therefore $2^{-n}$. By
rotational symmetry the same applies for the remaining quadrants,
including the special case of $\mathcal{K}_{11}$ for even
$n$. Extending $\mathcal{K}_n$ only decreases the error and the result
in Theorem~\ref{Thm:MajorityBound} continues to hold. The error
probability for the majority-based approach over all angles is
therefore bounded by $2/2^n$, which can likely be improved to
$1/2^{n}$. The approach requires $n$ samples in both the horizontal
and vertical direction therefore amounting to a total of $N=2n$
samples. In order to achieve an accuracy of $1/8$, we combine the
majority-based approach with a single stage of sign determination. The
resulting number of samples for different values of $\epsilon$ is
listed in Table~\ref{Table:MeasurementsFirstStage}.  A theoretical
bound on the number of samples can be found
using~\eqref{Eq:SignSamplesPi4}, giving
\[
N = 2\lceil \log_2(4/\epsilon)\rceil + \lceil 2\log_2(2/\epsilon)\rceil
\leq 9 + 4\log_2(1/\epsilon).
\]
This bound can be lowered by two samples if it can be shown that
$\alpha=0$ maximizes $f_n(\alpha)$ over $[0,\pi/2]$.

\begin{figure}
\centering
\includegraphics[width=0.55\textwidth]{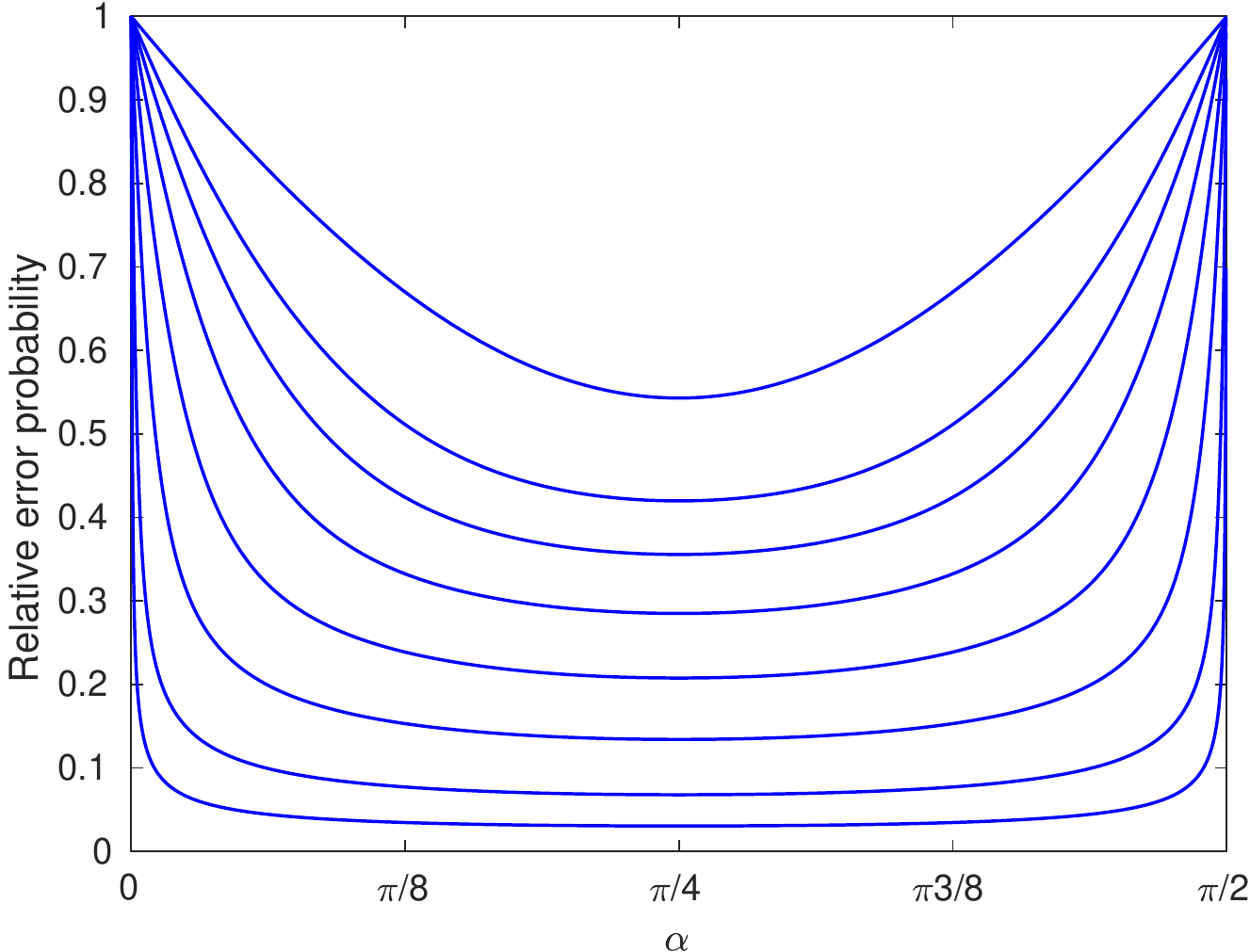}
\caption{Error probability $f_n(\alpha)$ scaled by $2^n$, with from top to bottom $n =
  1, 2, 3, 5, 10, 25, 100, 500$.}\label{Fig:MajorityError}
\end{figure}

\section{Evaluation of $N_{\epsilon}$}\label{Sec:Nepsilon}

For a given $\epsilon$ we can first determine $k_{\epsilon}$
using~\eqref{Eq:BoundK} and set $\bar{\epsilon} = \epsilon / k_{\epsilon}$.
Denote by $N_k$ the number of samples in steps
$k=1,\ldots,k_{\epsilon}-1$. For the first step we can either use the
triple-sign ($s$) or majority ($m$) based approaches giving respectively
\begin{equation}\label{Eq:N1}
N_1^s = 9 + 6\log_2(1/\bar{\epsilon}),\quad\mathrm{or}\quad
N_1^m = 9 + 4\log_2(1/\bar{\epsilon}).
\end{equation}
For the remaining steps we use the sign-based approach with angles
$\alpha = \pi/2^{k+1}$. Using Theorem~\ref{Thm:MeasurementsSign}, and
ignoring rounding up to the nearest integer we can take
\[
N_k = \frac{\log(1/\bar{\epsilon})}{\log(1/\sin(\pi/2^{k+1}))}
\leq \frac{\log(1 / \bar{\epsilon})}{\log(2^{k-1})} =
\frac{\log_2(1/\bar{\epsilon})}{k-1},
\]
where the inequality follows from $\sin(\pi/2^{k+1}) \leq  \pi /
2^{k+1}  \leq  4 / 2^{k+1} =  1 / 2^{k-1}$.
Summing over $N_k$ gives
\begin{eqnarray}
\sum_{k=2}^{k_{\epsilon}-1}N_k
&=&\log_2(1/\bar{\epsilon})\sum_{k=1}^{k_{\epsilon}-2}\frac{1}{k}\notag\\
&\leq&\log_2(1/\bar{\epsilon})\left(1 + \int_{1}^{k_{\epsilon}-2}
       x^{-1}dx\right) \notag\\
& = &\log_2(1/\bar{\epsilon})\left(1 + \log(k_{\epsilon}-2)\right)\label{Eq:Nk}
\end{eqnarray}
To account for rounding up of the intermediate values we add one for
each of the remaining $k_{\epsilon}-2$ steps. Combining \eqref{Eq:N1},
\eqref{Eq:Nk}, and the rounding term, and using $\log_2(1/\bar{\epsilon})
= \log_2(1/\epsilon) + \log_2(k_{\epsilon})$ gives
\begin{equation}\label{Eq:BoundNs}
N_{\epsilon}^s \leq 7 + k_{\epsilon} + (7 +
\log(k_{\epsilon}-2))\cdot(\log_2(1/\epsilon) + \log_2(k_{\epsilon}))
\end{equation}
for triple-sign based sampling and
\begin{equation}\label{Eq:BoundNm}
N_{\epsilon}^m \leq 7 + k_{\epsilon} + (5 +
\log(k_{\epsilon}-2))\cdot(\log_2(1/\epsilon) + \log_2(k_{\epsilon}))
\end{equation}
for majority-based sampling.

\begin{figure}
\centering
\begin{tabular}{cc}
\includegraphics[width=0.475\textwidth]{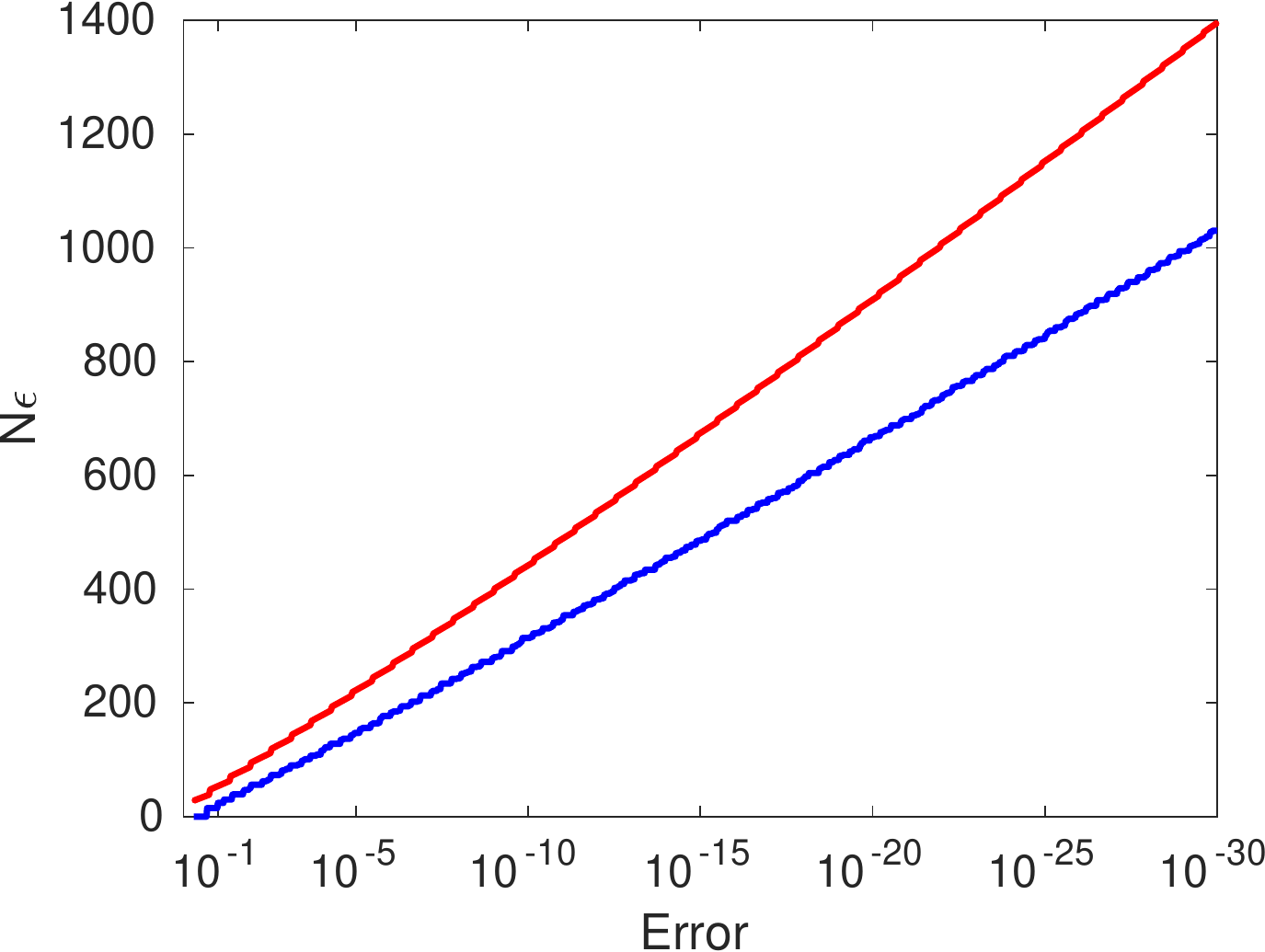}&
\includegraphics[width=0.475\textwidth]{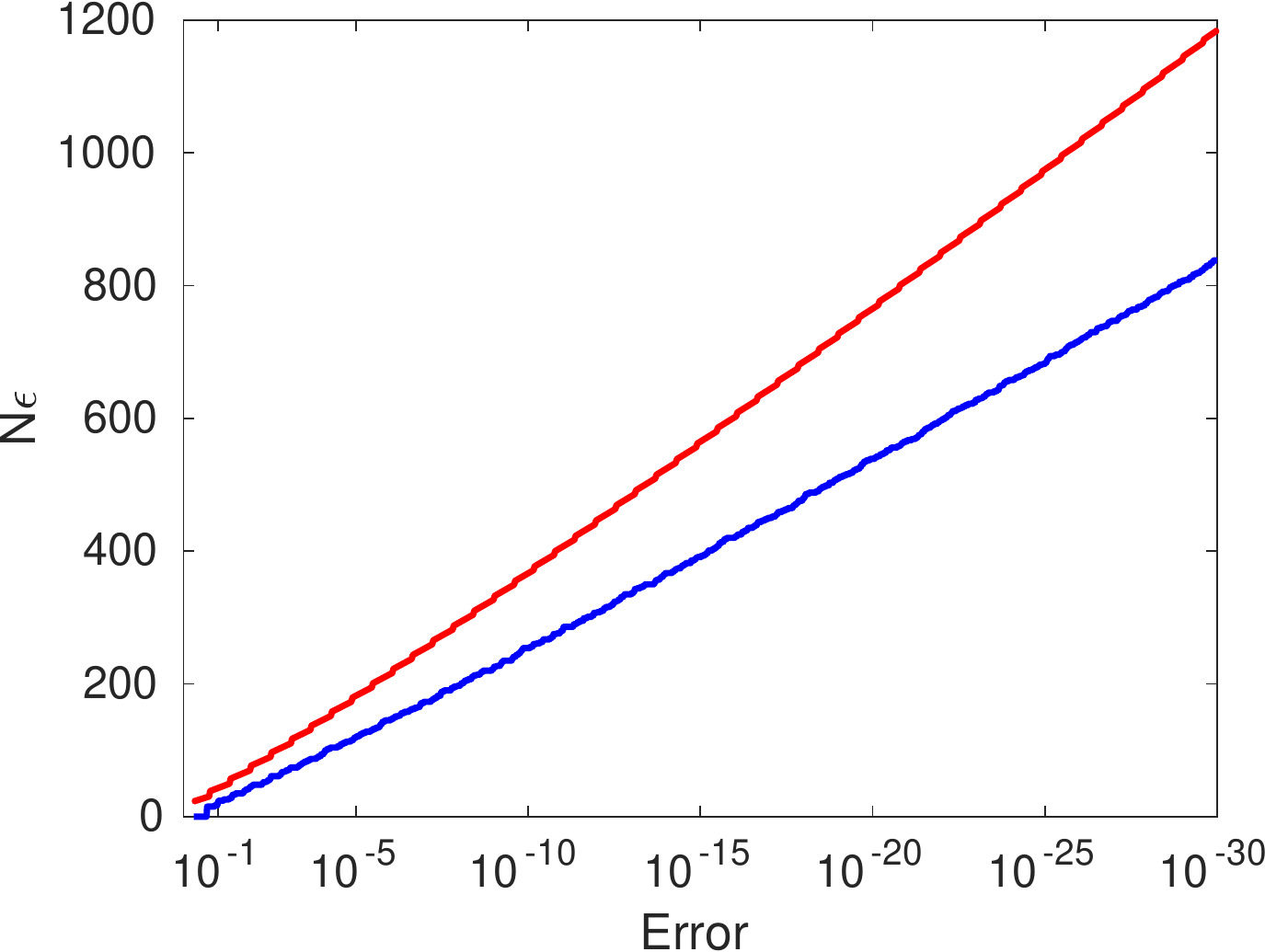}\\
({\bf{a}}) & ({\bf{b}})
\end{tabular}
\caption{Numerical evaluation of $N_{\epsilon}$ along with the
  theoretical upper bound, using (a) triple-sign
  sampling, and (b) majority-based sampling.}\label{Fig:PlotNepsilon}
\end{figure}

\paragraph{Numerical evaluation.}
For a numerical evaluation of $N_{\epsilon}$ we first determine the
critical iteration $k_{\epsilon}$ by finding the smallest integer $k$
that satisfies~\eqref{Eq:SufficientK}. Based on $k_{\epsilon}$ we set
$\bar{\epsilon} = \epsilon/k_{\epsilon}$ and use both the triple-sign
and majority-based sampling methods for the first iteration. After
that we use sign-based sampling with increasingly accurate phase
shifts to obtain the total number of evaluations before reaching
iteration $k_{\epsilon}$. The resulting values for $N_{\epsilon}$ are
plotted in Figure~\ref{Fig:PlotNepsilon} along with the theoretical
bounds given in equations~\eqref{Eq:BoundNs} and~\eqref{Eq:BoundNm}.
A summary of $k_{\epsilon}$ for different values of $\epsilon$ as well
as $N_{\epsilon}$ values for the two different sampling methods used
in the first iteration is given in
Table~\ref{Table:SummaryNk}. Finally, Table~\ref{Table:Mepsilon} gives
the total number of iterations needed to obtain a $2^{-(m+2)}$
accurate phase estimate with probability at least $1-\epsilon$ up to
and including iteration $k_{\epsilon}$. For each combination of
$\epsilon$ and $m$ that contains a number, we choose
$\bar{\epsilon} = \epsilon/m$. The dashed fields are in the regime
where a single measurement can be taken per additional bit of the
estimated angle, without having to change $\bar{\epsilon}$.

\begin{table}[!h]
\centering\setlength{\tabcolsep}{5.5pt}
\begin{tabular}{lrrrrrrrrrr}
\hline
Description $\backslash$ $\epsilon$ & $10^{-1}$ & $10^{-2}$ & $10^{-3}$ & $10^{-4}$ & $10^{-5}$ & $10^{-6}$ & $10^{-7}$ & $10^{-8}$ & $10^{-9}$ & $10^{-10}$\\
\hline
$k_{\epsilon}$ using \eqref{Eq:SufficientK} &    3 &    5 &    7 &    9 &   10 &   12 &   14 &   16 &   17 &   19\\
$k_{\epsilon}$ using \eqref{Eq:BoundK} &    4 &    6 &    7 &    9 &   11 &   12 &   14 &   16 &   17 &   19\\
$N_{\epsilon}^s$ (triple-sign) &   24 &   56 &   84 &  116 &  147 &  177 &  213 &  243 &  280 &  314\\
$N_{\epsilon}^s$ (triple-sign, bound~\eqref{Eq:BoundNs}) & {\it{  44}} & {\it{  84}} & {\it{ 123}} & {\it{ 163}} & {\it{ 197}} & {\it{ 237}} & {\it{ 277}} & {\it{ 317}} & {\it{ 353}} & {\it{ 394}}\\
$N_{\epsilon}^m$ (majority) &   24 &   48 &   72 &   96 &  121 &  147 &  175 &  199 &  226 &  256\\
$N_{\epsilon}^m$ (majority, bound~\eqref{Eq:BoundNm}) & {\it{  34}} & {\it{  66}} & {\it{  98}} & {\it{ 130}} & {\it{ 158}} & {\it{ 190}} & {\it{ 223}} & {\it{ 256}} & {\it{ 285}} & {\it{ 319}}\\
\hline
\end{tabular}
\caption{Summary of $k_{\epsilon}$ for different values of $\epsilon$
  as well as $N_{\epsilon}$ values for different sampling methods.}\label{Table:SummaryNk}
\bigskip\medskip
\centering\setlength{\tabcolsep}{3.5pt}
\begin{tabular}{lrrrrrrrrrrrrrrrrrrr}
\multicolumn{20}{l}{{Triple-sign sampling}}\\
\hline
$\epsilon$ $\backslash$ $m$ & $1$ & $2$ & $3$ & $4$ & $5$ & $6$ & $7$ & $8$ & $9$ & $10$ & $11$ & $12$ & $13$ & $14$ & $15$ & $16$ & $17$ & $18$ & $19$\\
\hline
$10^{-1}$ &   15 &   16 &   19 &   -- &   -- &   -- &   -- &   -- &   -- &   -- &   -- &   -- &   -- &   -- &   -- &   -- &   -- &   -- &   --\\
$10^{-2}$ &   33 &   36 &   41 &   42 &   45 &   -- &   -- &   -- &   -- &   -- &   -- &   -- &   -- &   -- &   -- &   -- &   -- &   -- &   --\\
$10^{-3}$ &   51 &   58 &   61 &   66 &   69 &   70 &   73 &   -- &   -- &   -- &   -- &   -- &   -- &   -- &   -- &   -- &   -- &   -- &   --\\
$10^{-4}$ &   69 &   78 &   83 &   86 &   89 &   94 &   97 &   98 &   99 &   -- &   -- &   -- &   -- &   -- &   -- &   -- &   -- &   -- &   --\\
$10^{-5}$ &   87 &   98 &  105 &  110 &  113 &  116 &  119 &  122 &  127 &  130 &   -- &   -- &   -- &   -- &   -- &   -- &   -- &   -- &   --\\
$10^{-6}$ &  105 &  118 &  125 &  132 &  137 &  140 &  145 &  150 &  153 &  156 &  159 &  160 &   -- &   -- &   -- &   -- &   -- &   -- &   --\\
$10^{-7}$ &  123 &  138 &  147 &  154 &  161 &  166 &  169 &  172 &  175 &  178 &  183 &  186 &  189 &  190 &   -- &   -- &   -- &   -- &   --\\
$10^{-8}$ &  141 &  158 &  169 &  178 &  185 &  190 &  195 &  198 &  203 &  206 &  209 &  212 &  215 &  218 &  219 &  220 &   -- &   -- &   --\\
$10^{-9}$ &  165 &  184 &  199 &  208 &  215 &  220 &  225 &  230 &  233 &  236 &  239 &  242 &  245 &  248 &  251 &  254 &  257 &   -- &   --\\
$10^{-10}$ &  183 &  206 &  219 &  228 &  235 &  242 &  247 &  252 &  259 &  264 &  267 &  272 &  275 &  278 &  281 &  284 &  287 &  290 &  291\\
\hline
\\
\multicolumn{20}{l}{{Majority-based sampling}}\\
\hline
$\epsilon$ $\backslash$ $m$ & $1$ & $2$ & $3$ & $4$ & $5$ & $6$ & $7$ & $8$ & $9$ & $10$ & $11$ & $12$ & $13$ & $14$ & $15$ & $16$ & $17$ & $18$ & $19$\\
\hline
$10^{-1}$ &   17 &   20 &   25 &   -- &   -- &   -- &   -- &   -- &   -- &   -- &   -- &   -- &   -- &   -- &   -- &   -- &   -- &   -- &   --\\
$10^{-2}$ &   29 &   34 &   43 &   44 &   49 &   -- &   -- &   -- &   -- &   -- &   -- &   -- &   -- &   -- &   -- &   -- &   -- &   -- &   --\\
$10^{-3}$ &   41 &   50 &   57 &   62 &   69 &   70 &   73 &   -- &   -- &   -- &   -- &   -- &   -- &   -- &   -- &   -- &   -- &   -- &   --\\
$10^{-4}$ &   55 &   68 &   73 &   80 &   83 &   88 &   93 &   96 &   97 &   -- &   -- &   -- &   -- &   -- &   -- &   -- &   -- &   -- &   --\\
$10^{-5}$ &   67 &   82 &   91 &   98 &  101 &  106 &  109 &  114 &  119 &  122 &   -- &   -- &   -- &   -- &   -- &   -- &   -- &   -- &   --\\
$10^{-6}$ &   79 &   96 &  107 &  114 &  121 &  124 &  131 &  136 &  141 &  144 &  147 &  148 &   -- &   -- &   -- &   -- &   -- &   -- &   --\\
$10^{-7}$ &   93 &  112 &  123 &  132 &  139 &  146 &  151 &  154 &  157 &  160 &  167 &  170 &  173 &  176 &   -- &   -- &   -- &   -- &   --\\
$10^{-8}$ &  105 &  126 &  141 &  150 &  159 &  166 &  171 &  174 &  181 &  184 &  189 &  192 &  195 &  198 &  199 &  200 &   -- &   -- &   --\\
$10^{-9}$ &  119 &  142 &  159 &  170 &  179 &  184 &  189 &  196 &  201 &  204 &  207 &  210 &  213 &  216 &  219 &  224 &  227 &   -- &   --\\
$10^{-10}$ &  133 &  160 &  173 &  186 &  193 &  200 &  209 &  214 &  221 &  226 &  229 &  234 &  239 &  244 &  247 &  250 &  253 &  256 &  257\\
\hline
\end{tabular}
\caption{Number of samples required to obtain an $2^{-(m+2)}$ accurate
  estimation of $\varphi$ with probability at least $1-\epsilon$ using
  triple-sign based sampling (top) and majority-based sampling
  (bottom). Dashed lines indicate the regime where a single extra
  measurement is needed for each successive $m$.}\label{Table:Mepsilon}
\end{table}

\section{Discussion}

In this work we have proposed and analyzed several sampling schemes
for use in quantum phase estimation based on Kitaev's algorithm, and
showed that using previous phase estimates to shift the phase can
reduce the number of measurements. Based on this we showed in
Section~\ref{Sec:PhaseShift} that we can obtain a theoretical sampling
complexity $N_{\epsilon} + m$ to obtain a $2^{-(m+2)}$ accurate
estimation of the phase $\varphi$ with probability at least
$\epsilon$. The proposed approach requires increasingly accurate
rotations, which may not be feasible in practice due to inherent
system noise or circuit complexity (see \cite{WIE2013Ka} for the
implementation of small rotations). Even with practical limitations on
the phase shift accuracy, as studied in more detail
in~\cite{AHM2012Ca}, the proposed sampling schemes can still reduce
the number of measurements, as shown, for example, in
Table~\ref{Table:Sign}.  From a theoretical point of view, having a
limited accuracy re-introduces a $\log(m)$ dependency in the
algorithmic complexity, and it will therefore be interesting to
analyze the application of the sampling schemes to the phase
estimation algorithm proposed in \cite{SVO2014HFa}.  Another potential
minor drawback of our approach is the dependency of each iteration
relies on the outcome of the previous one, thereby limiting the
potential parallelism to the independent measurements within each
iteration.

It remains to show that the maximum of $f_{n}(\alpha)$ in
\eqref{Eq:fn} over $[0,\pi/2]$ is attained at $\alpha = 0$. This would
confirm a sampling complexity of $2\log_2(1/\epsilon)$ for the
majority-based approach. This was verified for $n=1$ and $n=2$, and
Figure~\ref{Fig:MajorityError} strongly suggests this holds for all
$n$. Indeed, for $n=1$ we have
\[
f_1(\alpha) = (1 - \sin(\alpha) - \cos(\alpha) -
\sin(\alpha)\cos(\alpha)) / 4,
\]
which is convex over the given range due to concavity of the
trigonometric terms, and the result therefore follows from the
symmetry $f_n(\alpha) = f_n(\pi/2 - \alpha)$. Empirically, the error
functions for box-, wedge-, and majority-based sampling all exhibit
convexity or piecewise convexity. This may indicate a more general
relationship between the error over certain index sets $\mathcal{K}$
and $\alpha$.

\appendix
\section{Proof of Theorem~\ref{Thm:Delta}}\label{Sec:ProofThmDelta}

\begin{LabelTheorem}{\ref{Thm:Delta}}
  For any $0 \leq \eta \leq \pi/2$ we can compute an estimate
  $\tilde{\phi}$ of any $\phi \in [0,2\pi]$ with accuracy
  $\vert \tilde{\phi} - \phi\vert \leq \eta$ from sine and cosine
  estimates $\tilde{c}$ and $\tilde{s}$ with $\vert \tilde{c} -
  \cos(\phi)\vert \leq \delta$ and $\vert\tilde{s}
  -\sin(\phi)\vert \leq \delta$, whenever 
\begin{equation}
\delta \leq \delta(\phi) = \frac{\sin(\eta)}{\sqrt{2}}.
\end{equation}
For uniform estimation over $\phi$ this bound is tight.
\end{LabelTheorem}
\begin{proof}
  For $\eta=0$ the result holds trivially with $\delta = 0$, and we
  therefore only need to consider $\eta > 0$. We can recover any
  $\phi$ with accuracy $\eta$ from approximate sine and cosine values
  $\tilde{c}$ and $\tilde{s}$ if and only if $(\tilde{c},\tilde{s})$
  lies within a wedge of angles between $\phi-\eta$ and $\phi+\eta$
  (illustrated by the shaded region in Figure~\ref{Fig:Delta}).  For
  $\delta$, this means that the square with sides $2\delta$ centered
  on $\phi$ must to lie within the wedge. For $0 < \eta \leq \pi/4$ we
  can assume without loss of generality that $\phi \in [0,\pi/4]$. It
  can be seen that the intersection of the top-left corner of the box,
  at $(\cos(\phi)-\delta, \sin(\phi)+\delta)$, with the boundary of
  the wedge at angle $\phi+\eta$ determines the maximum value of
  $\delta$.  Formalizing, we write $\delta(\phi)$ to indicate the
  dependence on $\phi$ and denote the wedge boundary as
  $x = \alpha(\phi) y$, with
\[
\alpha(\phi) = \frac{\cos(\phi+\eta)}{\sin(\phi+\eta)}.
\]
For
  $\delta$ to be valid we need
  $\cos(\phi) - \delta \geq \alpha(\phi) (\sin(\phi) +
  \delta)$, which can be rewritten as
\begin{equation}\label{Eq:DeltaPhi}
\delta \leq \delta(\phi) = \frac{\cos(\phi) - \alpha(\phi)\sin(\phi)}{1+\alpha(\phi)}
\end{equation}
We then need to minimize $\delta(\phi)$ over the given range of
$\phi$ to find the largest value of $\delta$ that applies for all
$\phi$. Abbreviating $\alpha = \alpha(\phi)$ and
gradient $\alpha' = \alpha'(\phi)$, we have
\begin{eqnarray}
\delta'(\phi) &=& -\frac{\sin(\phi)}{1+\alpha} -
\frac{\alpha\cos(\phi)}{1+\alpha} -
\frac{\alpha'\cos(\phi)}{(1+\alpha)^2} -
\sin(\phi)\left(\frac{\alpha'}{1 + \alpha} -
                     \frac{\alpha\alpha'}{(1+\alpha)^2}\right)\notag\\
& = & -\sin(\phi)\left(\frac{1}{1+\alpha} +
      \frac{\alpha'}{(1+\alpha)^2}\right) -
      \cos(\phi)\left(\frac{\alpha}{1+\alpha} + \frac{\alpha'}{(1+\alpha)^2}\right)\label{Eq:DeltaPrime1}
\end{eqnarray}
From
\[
\alpha' = -\frac{1}{\sin^2(\phi -
  \eta)},\quad\mbox{and}\quad
\alpha^2 = \frac{\cos^2(\phi + \eta)}{\sin^2(\phi+\eta)} =
\frac{1 - \sin^2(\phi + \eta)}{\sin^2(\phi+\eta)} =
\frac{1}{\sin^2(\phi + \eta)} - 1,
\]
it follows that $\alpha' + \alpha^2 = -1$, or $\alpha' = -1
-\alpha^2$, which allows us to simplify the sine coefficient as
\begin{equation}\label{Eq:SinCoef}
\frac{1}{1+\alpha} +
      \frac{\alpha'}{(1+\alpha)^2} =\frac{1+\alpha}{(1+\alpha)^2} -
      \frac{1 + \alpha^2}{(1+\alpha)^2} = \frac{\alpha(1 - \alpha)}{(1+\alpha)^2},
\end{equation}
whereas for the cosine coefficient we find
\begin{equation}\label{Eq:CosCoef}
\frac{\alpha}{1+\alpha} + \frac{\alpha'}{(1+\alpha)^2} =
\frac{\alpha + \alpha^2 + \alpha'}{(1+\alpha)^2} = \frac{(\alpha - 1)}{(1+\alpha)^2}.
\end{equation}
Substituting \eqref{Eq:SinCoef} and \eqref{Eq:CosCoef} in
\eqref{Eq:DeltaPrime1} gives
\begin{equation}\label{Eq:DeltaPrime2}
\delta'(\phi) = \cos(\phi)\frac{(1 - \alpha)}{(1+\alpha)^2} -
\alpha\sin(\phi)\frac{(1 - \alpha)}{(1+\alpha)^2}.
\end{equation}
Noting that $\eta > 0$ and considering the range of $\phi$, we
have $0 < \sin(\phi + \eta) \leq 1$. This allows us to multiply the
first term in \eqref{Eq:DeltaPrime2} by
$\sin(\phi+\eta)/\sin(\phi+\eta)$, and expand the enumerator in
this term using the sum formula as
\[
\sin(\phi+\eta) = \sin(\phi)\cos(\eta) + \cos(\phi)\sin(\eta).
\]
Finally, expanding the enumerator $\cos(\phi+\eta)$ in the $\alpha$
term preceding $\sin(\phi)$ as
\[
\cos(\phi+\eta) = \cos(\phi)\cos(\eta) - \sin(\phi)\sin(\eta),
\]
and simplifying gives
\begin{equation}\label{Eq:DeltaPrime3}
\delta'(\phi) = \frac{(1 - \alpha)}{(1+\alpha)^2\sin(\phi + \eta)}\sin(\eta).
\end{equation}
All terms in this expression, except $\alpha-1$, are strictly
positive. The gradient is therefore zero only when $\alpha = 1$, which
happens at $\phi^* = \pi/4 - \eta$. For $\phi < \phi^*$ we have
$\alpha(\phi) > 1$ and therefore $\delta'(\phi) < 0$, whereas for
$\phi > \phi^*$ we have $\alpha(\phi) < 1$ and $\delta'(\phi) > 0$,
which shows that $\phi^*$ gives a minimizer. Evaluating
$\delta(\phi^*)$ in \eqref{Eq:DeltaPhi} and noting that
$\alpha(\phi^*) = 1$ then gives
\[
\delta \leq \delta(\phi) = (\cos(\pi/4 - \eta) - \sin(\pi/4 - \eta)) / 2.
\]
To obtain the desired result, we simplify $\delta(\phi)$ using the sum formulas and
$\cos(\pi/4)=\sin(\pi/4) = \sqrt{2}/2$:
\begin{eqnarray*}
\delta(\phi) & = & \frac{1}{2}\left((\cos(\pi/4)\cos(\eta) +
                   \sin(\pi/4)\sin(\eta)) - (\sin(\pi/4)\cos(\eta) -
                   \cos(\pi/4)\sin(\eta))\right)\\
& = & \frac{\sqrt{2}}{4}\left((\cos(\eta) + \sin(\eta)) -
      (\cos(\eta)-\sin(\eta))\right) = \frac{1}{\sqrt{2}}\sin(\eta).
\end{eqnarray*}
For $\pi/4\leq \eta \leq \pi/2$ we can assume without loss of
generality that $\phi \in [-\pi/4,0]$. In this case the top-left
corner of the box can again be seen to limit $\delta$. The argument as
given above follows through as is, thus completing the proof.
\end{proof}

\section{Proof of Theorem~\ref{Thm:BoxConvex}}\label{Sec:ProofThmBoxConvex}

\begin{LabelTheorem}{\ref{Thm:BoxConvex}}
  Choose $\delta > 0$ and let $f_n(p) = 1 - \mathrm{Pr}(X \in \mathcal{K}_{n,\delta}(p))$ with
  $\mathcal{K}_{n,p}$ as defined in \eqref{Eq:SetKnp}. Then for $n \geq \max\{1 + 1/\delta^2,3\}$,
$f_n(p)$ is piecewise convex on $[0,1]$ with breakpoints at
  $[0,1] \cap \{(k/n) \pm \delta/2\}_{k\in[n]}$.
\end{LabelTheorem}
\begin{proof}
  From the definition of $\mathcal{K}_{n,\delta}(p)$, it is clear that
  $\mathcal{K}_{n,\delta}(p)$ remains constant precisely on the
  (open) segment between the stated breakpoints. Choose any segment, then for
  all values of $p$ within this segment, the error is obtained by
  summing $B(k; n,p)$ over $k\not\in\mathcal{K}_{n,\delta}(p)$, with
\[
B(k; n,p) = {n \choose k}p^k(1-p)^{n-k}.
\]
In order to prove convexity of the error over the segment, we show
that the each of the terms $B(k; n,p)$ is convex in $p$ over the
segment.  For conciseness we normalize with respect to the binomial
coefficient and work with $B_{k,n}(p) := B(k; n,p)/ {n \choose k}$.
For $n=2$, observe that the second derivative $B_{1,2}''(p) = -2$ is
negative, which means that $B_{1,2}(p)$ is concave. We therefore
require that $n \geq 3$. For $k=0$ and $k=n$ we find
\[
B_{0,n}''(p) = n(n-1)(1-p)^{n-2},\quad\mathrm{and}\quad
B_{n,n}''(p) = n(n-1)p^{n-2}.
\]
The second derivatives are nonnegative over the domain $p\in[0,1]$ and
the functions are therefore convex.
For $0 < k < n$ we have
\begin{eqnarray}
B_{k,n}'(p) &=& (k(1-p) - (n-k)(p))(p^{k-1}(1-p)^{n-k-1})\notag\\
&=& (k - np)\left(p^{k-1}(1-p)^{n-k-1}\right),\label{Eq:GradBkn}
\end{eqnarray}
and the gradient reaches zero when $p=0$, $p=1$, or $p = k/n$. For
$k=1$ we find
\begin{eqnarray*}
B_{1,n}''(p) & = & \left[-n(1-p) - (1-np)(n-2)\right](1-p)^{n-3}\\
& = & \left[np(n-1) - 2(n-1)\right](1-p)^{n-3}
\end{eqnarray*}
For convexity we want $B_{1,n}''(p) \geq 0$, and therefore require
that the square-bracketed term be nonnegative. Solving for $p$ then
gives convexity of $B_{1,n}(p)$ for $p \geq 2/n$. By symmetry, it
follows that for $k=n-1$, $B_{{n-1},n}(p)$ is convex for $p \leq 1 - 2/n$.
Finally, for $2 \leq k \leq n-2$ it follows from \eqref{Eq:GradBkn} that
\begin{eqnarray*}
B''_{k,n}(p) & = &
\left(-np(1-p) + (k-np)\left((k-1)(1-p) -
                   (n-k-1)p\right)\right)\cdot\left(p^{k-2}(1-p)^{n-k-2}\right)\\
& = & \left[n(n-1)p^2 - 2k(n-1)p + k(k-1) \right]\cdot\left(p^{k-2}(1-p)^{n-k-2}\right).
\end{eqnarray*}
The term in square brackets is a quadratic in $p$, and solving for the
roots gives
\[
p_{\pm} = \frac{k}{n}\pm \frac{\sqrt{k(n-1)(n-k)}}{n(n-1)}.
\]
The deviation is maximum at $k=n/2$, which gives
\[
\frac{\sqrt{k(n-1)(n-k)}}{n(n-1)} \leq \frac{(n/2)\sqrt{n-1}}{n(n-1)}
= \frac{1}{2\sqrt{n-1}}.
\]
The second derivative $B_{k,n}''$ is therefore guaranteed to be
nonnegative, and $B_{k,n}$ convex, when $p$ is at least
$1/(2\sqrt{n-1})$ away from the maximum at $k/n$. It can be verified
that the same sufficient condition applies for $k=0$ and $k=n$.

For any $p$ in the selected segment we know that
$\mathcal{K}_{n,\delta}(p)$ remains constant and that
$\abs{k/n - p} \geq \delta/2$ for any
$k\not\in\mathcal{K}_{n,\delta}(p)$. To guarantee convexity we
therefore require that
\[
\delta/2 \geq 1/(2\sqrt{n-1}),
\]
which simplifies to $n \geq 1 + 1/\delta^2$.
\end{proof}

\section{Proof of Theorem~\ref{Thm:MajorityBound}}\label{Sec:ProofThmMajorityBound}
\begin{LabelTheorem}{\ref{Thm:MajorityBound}}
Let $\mathcal{K}_n = \{(i,j) \mid i,j \in [0,n], j \geq n-i+1\}$, then
for all $\alpha \in [0,\pi/2]$
\[
1 - \mathrm{Pr}(\mathcal{K}_n \mid n,\alpha) \leq \frac{2}{2^n}.
\]
\end{LabelTheorem}
\begin{proof}
Denote by $\mathcal{E}_n = \{(i,j) \mid i,j \in [0,n], i+j \leq n\}$
the complement of $\mathcal{K}_n$. The error probability
$\mathrm{Pr}(\mathcal{E}_n\mid n, \alpha) = 1 -
\mathrm{Pr}(\mathcal{K}_n \mid n, \alpha)$ is then
obtained by summing $f_{i,j}$ over $(i,j) \in\mathcal{E}_n$, where
\[
f_{i,j}(\alpha) = {n\choose i}{n \choose j}p_x^i(1-p_x)^{n-i}p_y^j(1-p_y)^{n-j}.
\]
Defining the diagonal sums $k\in[0,n]$ as
\[
d_k(\alpha) = \sum_{i=0}^{k} f_{i,k-j}(\alpha),
\]
we can equivalently write
$\mathrm{Pr}(\mathcal{E}_n\mid n, \alpha) = \sum_{k=0}^nd_k(\alpha)$.
For $\alpha \in [0,\pi/2]$ it is easily seen that
$d_k(\alpha) = d_k(\pi/2 - \alpha)$, and it therefore suffices to show
the desired result for $\alpha \in [0,\pi/4]$. As a first step, we
bound the value of the main diagonal $d_n$ by $2^{-n}$:
\begin{eqnarray}
d_n(\alpha) & = & \sum_{j=0}^{n}{n\choose n-j}{n\choose
          j}p_x^{n-j}(1-p_x)^{j}p_y^j(1-p_y)^{n-j}\notag\\
& = & (p_x(1-p_y))^n \sum_{j=0}^{n}{n\choose j}^2
      \left(\frac{p_y(1-p_x)}{p_x(1-p_y)}\right)^j\notag\\
& = & (p_x(1-p_y))^n \sum_{j=0}^{n}\left({n\choose j}
      \left(\frac{p_y(1-p_x)}{p_x(1-p_y)}\right)^{j/2}\right)^2\notag\\
& \leq & (p_x(1-p_y))^n \left(\sum_{j=0}^{n}{n\choose j}
      \left(\frac{\sqrt{p_y(1-p_x)}}{\sqrt{p_x(1-p_y)}}\right)^j\right)^2\label{Eq:BoundDn}\\
& \stackrel{(i)}{=} &(p_x(1-p_y))^n \left(1 + \frac{\sqrt{p_y(1-p_x)}}{\sqrt{p_x(1-p_y)}}\right)^{2n}\notag\\
& \stackrel{(ii)}{=} &(p_x(1-p_y))^n \left(\frac{1}{2p_x(1-p_y)}\right)^n\notag\\
& =& 2^{-n},\notag
\end{eqnarray}
where (i) uses the binomial theorem and (ii) follows from the observation that
\begin{eqnarray*}
\left(1 + \frac{\sqrt{p_y(1-p_x)}}{\sqrt{p_x(1-p_y)}}\right)^{2}
& = &  \left(\frac{\sqrt{p_x(1-p_y)} + \sqrt{p_y(1-p_x)}}{\sqrt{p_x(1-p_y)}}\right)^{2}\\
& = & \frac{p_x(1-p_y) + 2\sqrt{p_x(1-p_x)p_y(1-p_y)} + p_y(1-p_x)}{p_x(1-p_y)}\\
& \stackrel{\eqref{Eq:4p(1-p)}}{=} & \frac{p_x(1-p_y) + 2\sqrt{\sin^2(\alpha)\cos^2(\alpha)/16} + p_y(1-p_x)}{p_x(1-p_y)}\\
& = & \frac{(1+\cos(\alpha)(1-\sin(\alpha))/4 + \sin(\alpha)\cos(\alpha)/2 + (1+\sin(\alpha))(1-\cos(\alpha))/4}{p_x(1-p_y)}\\
& = & \frac{1}{2p_x(1-p_y)}.\\
\end{eqnarray*}

\begin{figure}[t]
\centering
\setlength{\tabcolsep}{6pt}
\begin{tabular}{cccc}
\includegraphics[width=0.21\textwidth]{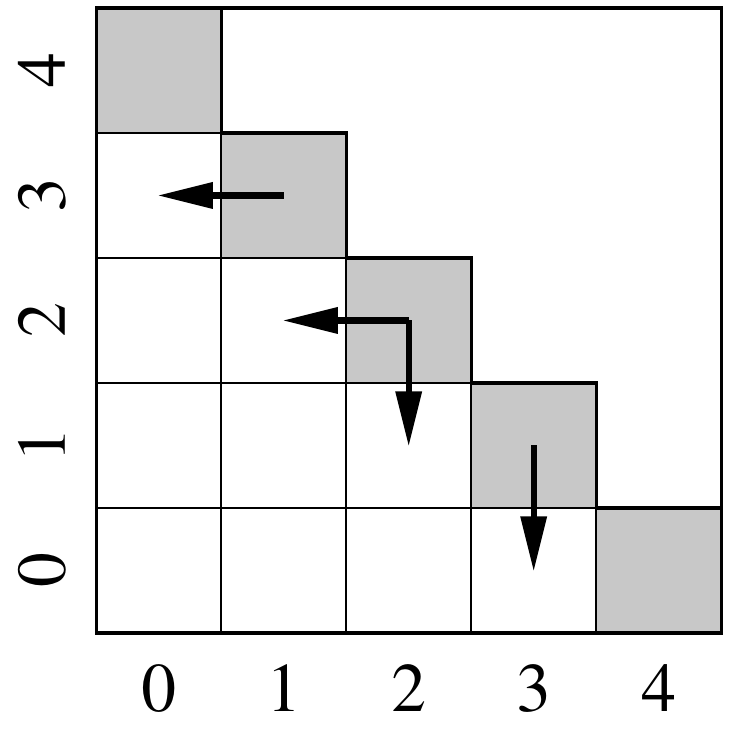}&
\includegraphics[width=0.21\textwidth]{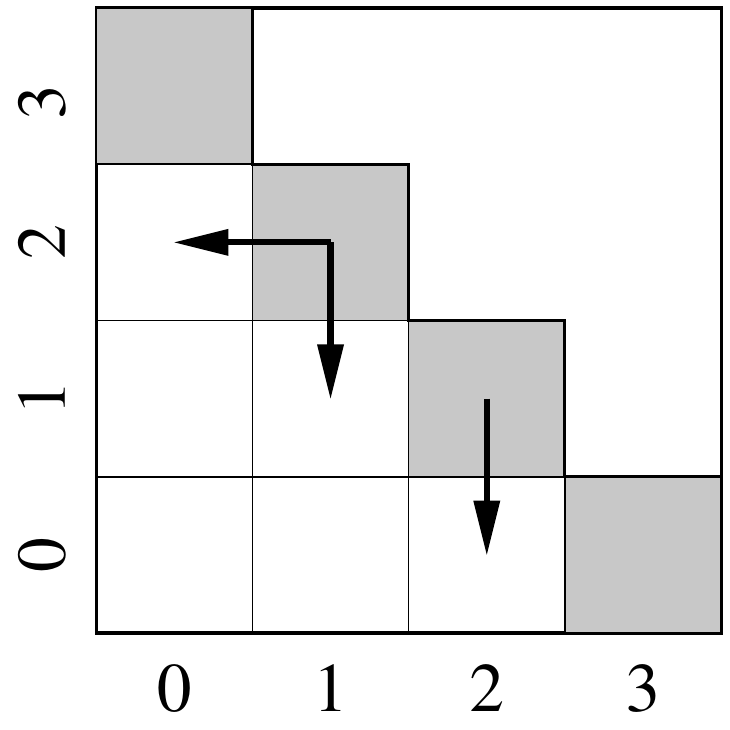}&
\includegraphics[width=0.21\textwidth]{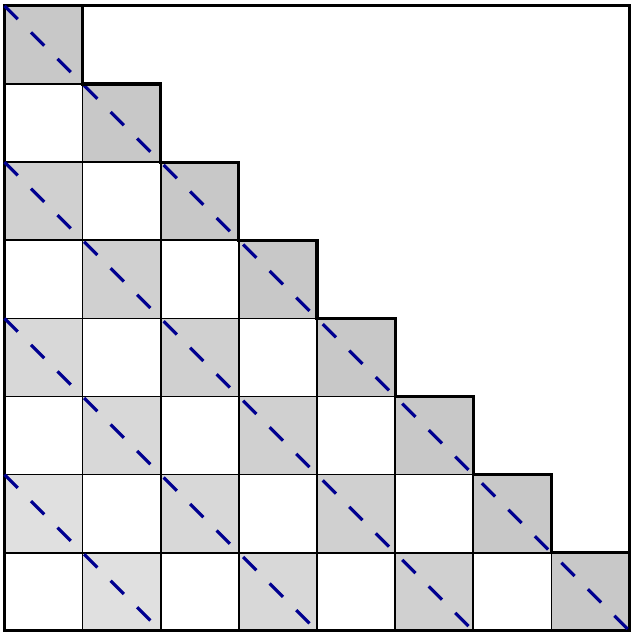}&
\includegraphics[width=0.21\textwidth]{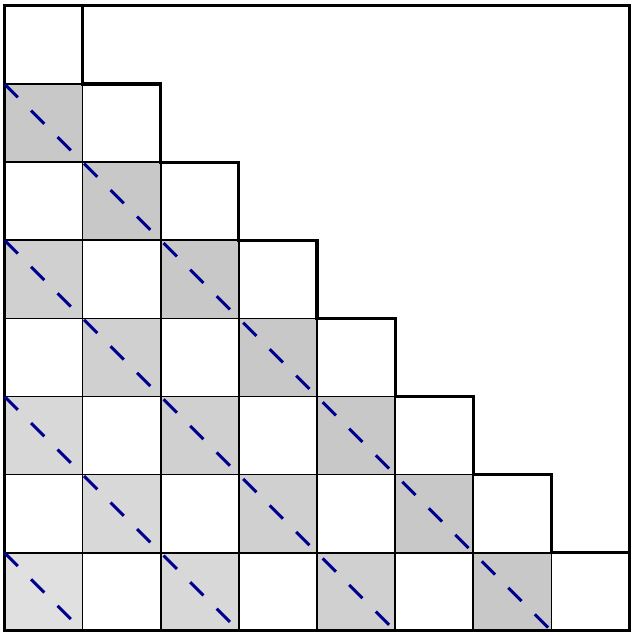}\\
({\bf{a}}) & ({\bf{b}}) & ({\bf{c}}) & ({\bf{d}})
\end{tabular}
\caption{Constructions for bounding (a,b) $d_{n-1}$ based on
  $d_{n}$, and (c,d) the sum of $d_k$ for strided $k$.}\label{Fig:GridMajority}
\end{figure}

For the second step we derive a bound on $d_{n-1}$ based on $d_n$, from
which we then obtain a bound on $d_{n-1} + d_n$. For $i \geq 1$ we
have
\[
f_{i-1,j} = f_{i,j}\cdot\frac{1-p_x}{p_x}\cdot\frac{i}{n-i+1}.
\]
The right-most term, which accounts for the change in the binomial
coefficient ${n\choose i}$, is less than or equal to 1 for $i
\leq n/2$ when $n$ is even, and for $i \leq (n+1)/2$ when $n$ is odd.
A similar argument applies for the transition from $f_{i,j}$ to
$f_{i,j-1}$ for $1 \leq j \leq (n+1)/2$, allowing us to bound the elements
on the $(n-1)$-diagonal $d_{n-1}$ as follows:
\[
f_{i,j} \leq \begin{cases}
\displaystyle \frac{1 - p_x}{p_x}f_{i+1,j} & i < (n-1)/2,\\[14pt]
\displaystyle \frac{1 - p_y}{p_y}f_{i,j+1} & \mathrm{otherwise}.
\end{cases}
\]
As illustrated in Figures~\ref{Fig:GridMajority}(a) and (b), this
approach uses the middle element of the main diagonal twice. Taking
this into account, and effectively doing the same for all elements, we
have
\begin{equation}\label{Eq:BoundDn-1}
d_{n-1} \leq \left(\frac{1-p_x}{p_x} + \frac{1-p_y}{p_y}\right)d_n =
\frac{p_y(1-p_x) + p_x(1-p_y)}{p_xp_y}d_n.
\end{equation}
Combining \eqref{Eq:BoundDn} and \eqref{Eq:BoundDn-1} we have
\begin{equation}\label{Eq:BoundSumDn}
d_n + d_{n-1} \leq \frac{p_x + p_y - p_xp_y}{p_xp_y}\cdot 2^{-n}.
\end{equation}

As the third step, we derived bound on $d_{k-2}$ based on
$d_{k}$. Consider any diagonal $2\leq k \leq n$, with $0 < i < k$ and
$j = k-i$, then
\[
{n\choose i-1}{n\choose j-1} = {n\choose i-1}{n\choose k-i-1} = \frac{i}{n-i+1}\cdot\frac{k-i}{n-k+i+1}{n\choose i}{n\choose k-i}
= \kappa {n\choose i}{n\choose j}
\]
Since $k \leq n$, the multiplicative term $\kappa$ satisfies
\[
\kappa = \frac{i(k-i)}{i(k-i) + n(n-k+1) - k+2} \leq \frac{i(k-i)}{i(k-i) + n -
  k+2} \leq \frac{i(k-i)}{i(k-i) +2} < 1.
\]
It therefore follows that
\[
f_{i-1,k-i-1} \leq f_{i,k-i} \cdot \frac{(1 - p_x)}{p_x}\cdot\frac{(1-p_y)}{p_y}.
\]
The transition from diagonal $k$ to $k-2$ follows by summing over all
elements $i+j = k-2$, giving
\[
d_{k-2} \leq \frac{(1 - p_x)(1 - p_y)}{p_xp_y} d_{k} = \tau d_{k},
\]
with $\tau < 1$, as shown in Figure~\ref{Fig:Alpha}(a). As a fourth
step we sum over the even and odd diagonals. Starting at $k = n$ or
$k = n-1$ we have
\[
\sum_{i=0}^{k/2} d_{k-2i}\leq
\sum_{i=0}^{\infty}\tau^id_k = \frac{1}{1 - \tau}d_k = \frac{p_xp_y}{p_x +
  p_y - 1}\cdot d_k
\]
For the sum of the diagonals, and hence that $f_{i,j}$ over the error set set $\mathcal{E}_n$,
it follows from \eqref{Eq:BoundSumDn} that
\[
\sum_{k=0}^{n} d_k \leq
\frac{p_xp_y}{p_x + p_y - 1}(d_n + d_{n-1})
\leq \frac{p_xp_y}{p_x + p_y - 1}\cdot\frac{p_x + p_y - p_xp_y}{p_xp_y}\cdot 2^{-n}
= \frac{p_x + p_y - p_xp_y}{p_x + p_y - 1}\cdot 2^{-n}.
\]
The desired result then follows from the observation that $(p_x + p_y -
p_xp_y)/ (p_x + p_y - 1) \leq 2$, as illustrated in Figure~\ref{Fig:Alpha}(b).
\end{proof}
\begin{figure}
\centering
\begin{tabular}{cc}
\includegraphics[height=140pt]{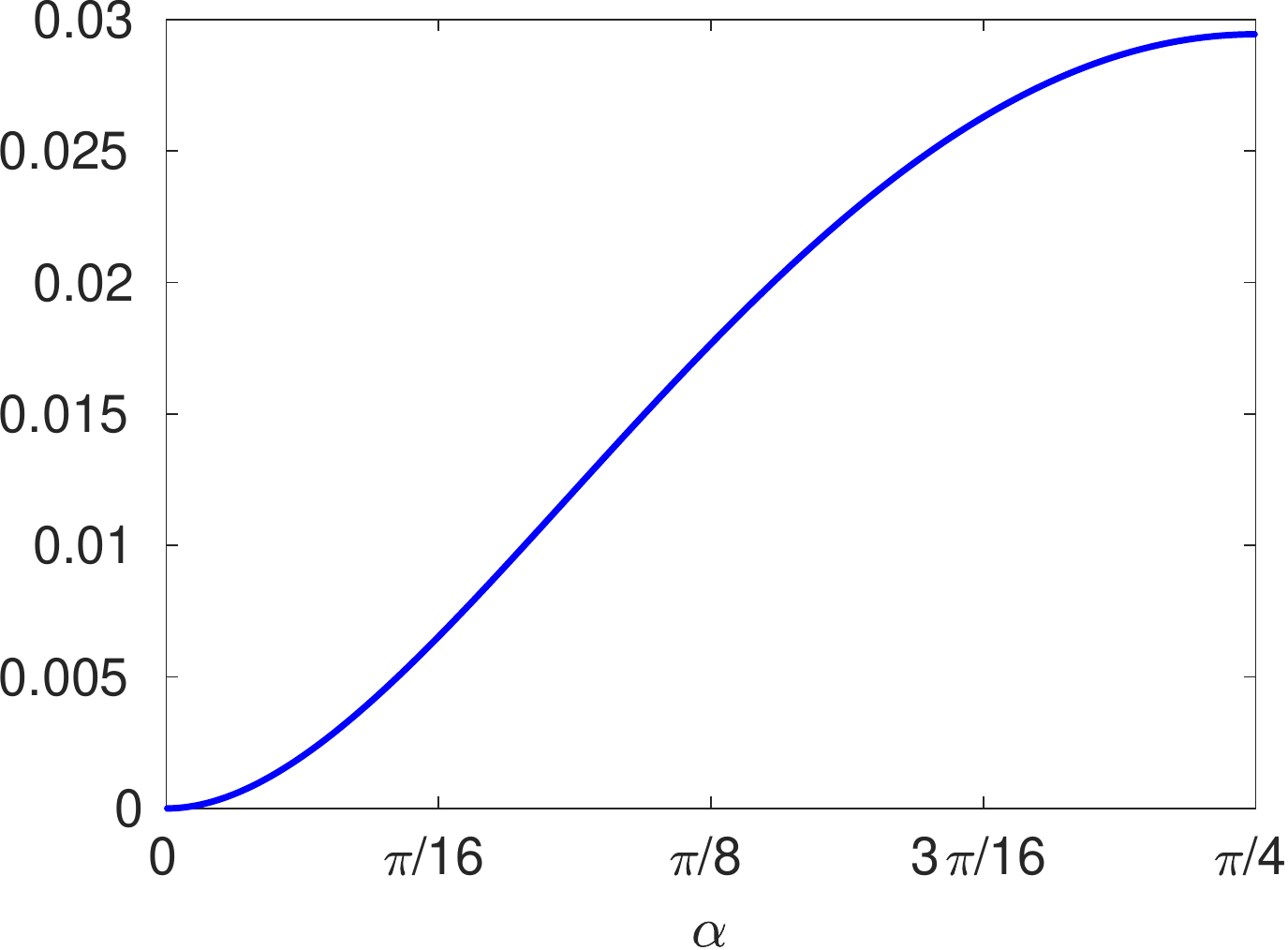} &
\includegraphics[height=140pt]{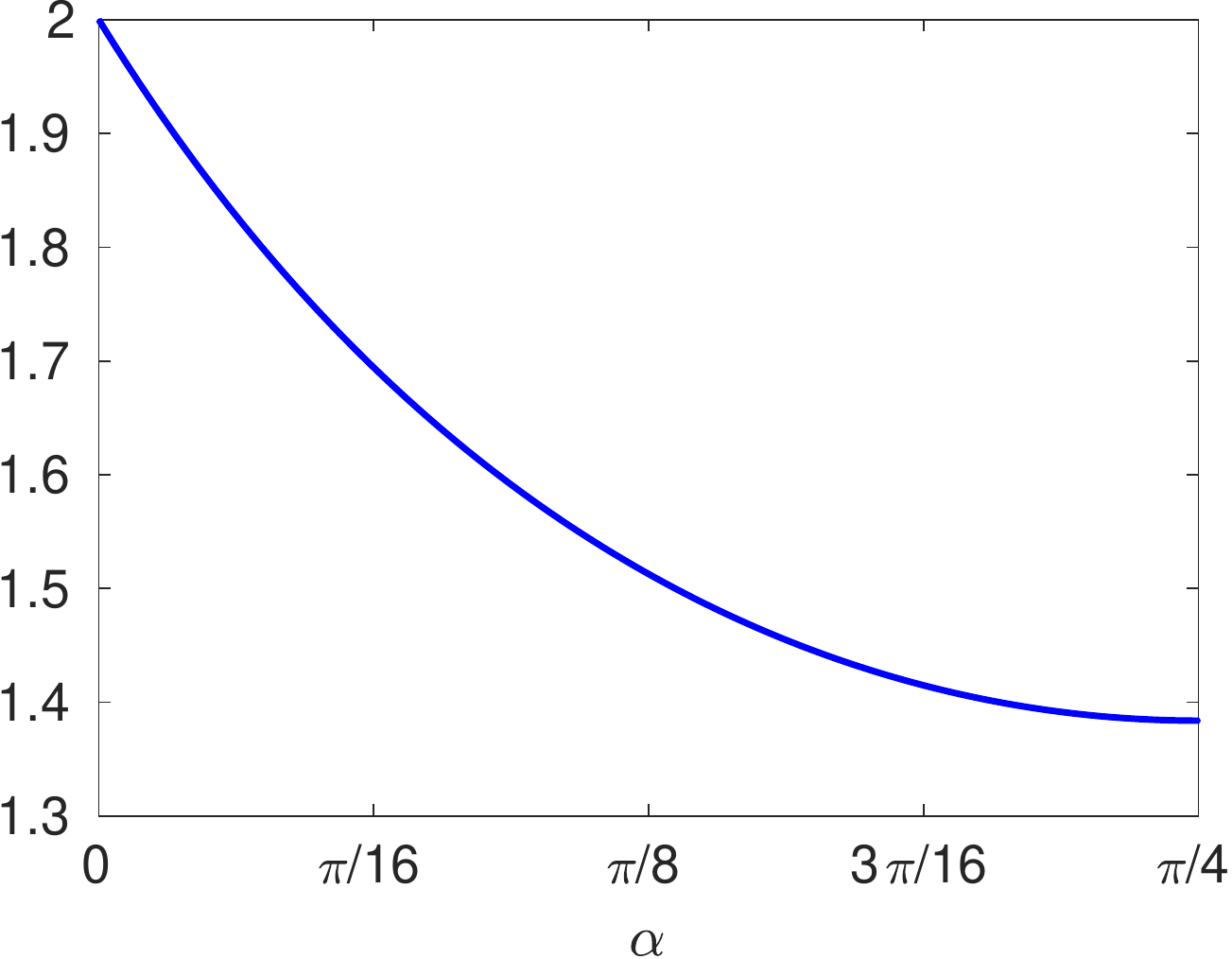} \\
({\bf{a}}) $\frac{(1-p_x)(1-p_y)}{p_xp_y}$ &
({\bf{b}}) $\frac{p_x + p_y - p_xp_y}{p_x + p_y - 1}$
\end{tabular}
\caption{Plots of key quantities used in the proof of Theorem~\ref{Thm:MajorityBound}.}\label{Fig:Alpha}
\end{figure}

\bibliography{bibliography}

\begin{thebibliography}{10}

\bibitem{AHM2012Ca}
Hamed Ahmadi and Chen-Fu Chiang.
\newblock Quantum phase estimation with arbitrary constant-precision phase
  shift operators.
\newblock {\em Quantum Information \& Computation}, 12(9\&10):0854--0875, 2012.

\bibitem{ARR1989Ga}
Richard Arratia and Louis Gordon.
\newblock Tutorial on large deviations for the binomial distribution.
\newblock {\em Bulletin of Mathematical Biology}, 51(1):125--131, 1989.

\bibitem{ASP2005DLHa}
Al{\'a}n Aspuru-Guzik, Anthony~D. Dutoi, Peter~J. Love, and Martin Head-Gordon.
\newblock Simulated quantum computation of molecular energies.
\newblock {\em Science}, 309(5741):1704--1707, 2005.

\bibitem{CHI2000PRa}
Andrew~M. Childs, John Preskill, and Joseph Renes.
\newblock Quantum information and precision measurement.
\newblock {\em Journal of Modern Optics}, 47(2/3):155--176, 2000.

\bibitem{CLE1998EMMa}
Richard Cleve, Arthur Ekert, Chiara Macchiavello, and Michele Mosca.
\newblock Quantum algorithms revisited.
\newblock {\em Proceedings of the Royal Society A}, 454(1969):339--354, 1998.

\bibitem{DOB2007JSWa}
Miroslav Dob\v{s}\'{\i}\v{c}ek, G\"{o}ran Johansson, Vitaly Shumeiko, and
  G\"{o}ran Wendin.
\newblock Arbitrary accuracy iterative phase estimation algorithm as a two
  qubit benchmark.
\newblock {\em Physical Review A}, 76(3):030306, 2007.

\bibitem{KIM2015LYa}
Shelby Kimmel, Guang~Hao Low, and Theodore~J. Yoder.
\newblock Robust calibration of a universal single-qubit gate set via robust
  phase estimation.
\newblock {\em Physical Review A}, 92(6):062315, 2015.

\bibitem{KIT1995a}
Alexei~{Yu}. Kitaev.
\newblock Quantum measurements and the {A}belian stabilizer problem.
\newblock arXiv preprint quant-ph/9511026, 1995.
\newblock (See also Electronic Colloquium on Computational Complexity,
  TR96-003, 1996).

\bibitem{KIT2002SVa}
Alexei~{Yu}. Kitaev, Alexander~H. Shen, and Mikhail~N. Vyalyi.
\newblock {\em Classical and Quantum Computation}.
\newblock American Mathematical Society, 2002.

\bibitem{KNI2007OSa}
Emmanuel Knill, Gerardo Ortiz, and Rolando~D. Somma.
\newblock Optimal quantum measurements of expectation values of observables.
\newblock {\em Physical Review A}, 75(1):012328, 2007.

\bibitem{NIE2010Ca}
Michael~A. Nielsen and Isaac~L. Chuang.
\newblock {\em Quantum Computation and Quantum Information}.
\newblock Cambridge University Press, 2010.

\bibitem{OMA2016BKRa}
Peter J.~J. O'Malley, Ryan Babbush, Ian~D. Kivlichan, Jonathan Romero,
  Jarrod~R. McClean, Rami Barends, Julian Kelly, Pedram Roushan, Andrew
  Tranter, Nan Ding, Brooks Campbell, Yu~Chen, Zijun Chen, Ben Chiaro, Andrew
  Dunsworth, Austin~G. Fowler, Evan Jeffrey, Erik Lucero, Anthony Megrant,
  Josh~Y. Mutus, Matthew Neeley, Charles Neill, Chris Quintana, Daniel Sank,
  Amit Vainsencher, James Wenner, Ted~C. White, Peter~V. Coveney, Peter~J.
  Love, Hartmut Neven, Al{\'a}n Aspuru-Guzik, and John~M. Martinis.
\newblock Scalable quantum simulation of molecular energies.
\newblock {\em Physical Review X}, 6(3):031007, 2016.

\bibitem{SHO1997a}
Peter~W. Shor.
\newblock Polynomial-time algorithms for prime factorization and discrete
  logarithms on a quantum computer.
\newblock {\em SIAM Journal on Computing}, 26(5):1484--1509, 1997.

\bibitem{SVO2014HFa}
Krysta~M. Svore, Matthew~B. Hastings, and Michael Freedman.
\newblock Faster phase estimation.
\newblock {\em Quantum Information \& Computation}, 14(3--4):306--328, March
  2014.

\bibitem{TEM2011OVPa}
Kristan Temme, Tobias~J. Osborne, Karl~G. Vollbrecht, David Poulin, and Frank
  Verstraete.
\newblock Quantum {M}etropolis sampling.
\newblock {\em Nature}, 471:87--90, 2011.

\bibitem{WHI2011BAa}
James~D. Whitfiled, Jacob Biamonte, and Al{\'a}n Aspuru-Guzik.
\newblock Simulation of electronic structure {H}amiltonians using quantum
  computers.
\newblock {\em Molecular Physics}, 109(5):735--750, 2011.

\bibitem{WIE2013Ka}
Nathan Wiebe and Vadym Kliuchnikov.
\newblock Floating point representations in quantum circuit synthesis.
\newblock {\em New Journal of Physics}, 13:093041, 2013.

\end{thebibliography}

\end{document}